\documentclass[12pt,reqno]{amsart}
\usepackage{cite}
\usepackage{amsfonts}
\usepackage{graphicx}
\usepackage{graphics}
\usepackage{float}
\usepackage{amscd}
\usepackage[latin1]{inputenc}
\usepackage[english]{babel}
\usepackage{amsmath}
\usepackage{amssymb}
\usepackage{epsfig}
\usepackage[usenames]{color}
\usepackage{caption}
\usepackage{subfigure}
\usepackage{amsmath}

\providecommand{\U}[1]{\protect\rule{.1in}{.1in}}
\providecommand{\U}[1]{\protect\rule{.1in}{.1in}}
\allowdisplaybreaks
\newtheorem{teo}{Theorem}
\theoremstyle{plain}
\newtheorem{acknowledgement}{Acknowledgement}

\newtheorem{corollary}{Corollary}

\newtheorem{lem}{Lemma}

\newtheorem{proposition}{Proposition}
\newtheorem{remark}{Remark}

\numberwithin{equation}{section}

\DeclareMathOperator{\senh}{senh}
\DeclareMathOperator{\csch}{csch}
\DeclareMathOperator{\sen}{sen}
\DeclareMathOperator{\sd}{sd}

\DeclareMathOperator{\sech}{sech}

\DeclareMathOperator{\erf}{erf}

\begin{document}
\title[variable coefficient reaction-diffusion equations]{ Riccati-Ermakov
systems and explicit solutions for variable coefficient reaction-diffusion
equations}
\author[E. Pereira]{Enrique Pereira}
\address[E. Pereira]{Department of Mathematical Sciences, Southern Methodist
University and Institute of Applied Mathematics, University of Cartagena .}
\email{epereirabatista@mail.smu.edu}
\author[E. Suazo]{Erwin Suazo}
\address[E. Suazo]{School of Mathematical and Statistical Sciences,
University of Texas, Rio Grande Valley, 1201 West University Dr. Edinburg,
Texas 78539-2999}
\email{erwin.suazo@utrgv.edu}
\author[J. Trespalacios]{Jessica Trespalacios}
\address[J. Trespalacios]{Department of Mathematical Sciences, Universidad
Santiago de Cali, Santiago de Cali, Colombia.}
\email{jessica.trespalacios00@usc.edu.co}
\date{\today }

\begin{abstract}
We present several families of nonlinear reaction diffusion equations with
variable coefficients including generalizations of Fisher-KPP and Burgers
type equations. Special exact solutions such as traveling wave, rational,
triangular wave and N-wave type solutions are shown. By means of similarity
transformations the variable coefficients are conditioned to satisfy Riccati
or Ermakov systems of equations. When the Riccati system is used, conditions
are established so that finite-time singularities might occur. We explore
solution dynamics across multi-parameters. In the suplementary material, we
provide a computer algebra verification of the solutions and exemplify
nontrivial dynamics of the solutions.

\noindent{\footnotesize {\textbf{Keywords and Phrases}. \textit{Similarity
tranformations, variable coefficient Burgers equation, variable coefficient
Fisher-KPP equation, Riccati-Ermakov systems of ODEs, exact solutions,
multiparameters.}}\newline
} 
\end{abstract}

\maketitle


\section{Introduction}

Most physical and biological systems are not homogeneous, in part due to
fluctuations in environmental conditions and the presence of nonuniform
media. Therefore, most of the nonlinear equations with real applications
possess coefficients varying spatially and/or temporally. Reaction-diffusion
equations play a fundamental role in a large number of models of heat
diffusion and reaction processes in nonlinear acoustics \cite{Crighton},
biology, chemistry, genetics and many other areas of research \cite{Abl}, 
\cite{Agr}, \cite{Caze:her}, \cite{Debnath05}, \cite{Polyanin12}, \cite%
{Fisher37} and \cite{Brezis86}. In this paper we obtain explicit solutions
for the general initial value problems of a generalized Fisher-KPP with
variable coefficients given by {\small {%
\begin{equation}
\dfrac{\partial u}{\partial t}=a(t)\dfrac{\partial ^{2}u}{\partial x^{2}}%
-(g(t)-c(t)x)\dfrac{\partial u}{\partial x}%
+(d(t)+L(t)+M(t)x-B(t)x^{2})u+h(x,t)|u|^{p}u,  \tag{GNLH}
\end{equation}%
}}and a generalized Burgers equation (GBE) with variable coefficients 
\begin{equation}
v_{t}+4a(t)(vv_{x}-v_{xx})=-b(t)x+f(t),  \tag{GBE}  \label{GBE0}
\end{equation}%
{\small \ }where $a(t),$ $b(t),$ $B(t)$, $c(t),$ $d(t),$ $M(t),$ $L(t),$ $%
f(t)$ and $g(t)$ are suitable functions that depend only on the time
variable and $h(x,t)$ depends on $t>0,$ $x\in 
\mathbb{R}
$ and $p>0$. If in \ref{GNLH} we choose $d(t)$ constant $L(t)=h(x,t)=1$ and $%
g(t)=c(t)=M(t)=B(t)=0,$ we obtain the Fisher-KPP equation \cite{Briton}, 
\cite{Fisher37}, \cite{Kolmogorov37}, \cite{Fife} and \cite{Murray}; if we
take the same coefficients with $d(t)=0,$ we obtain the nonlinear heat
equation with absorption. If in \ref{GBE0} we consider $a(t)=1$ and $b(t)$=$%
f(t)=0,$ we obtain the standard Burgers equation \cite{Burgers1}, \cite%
{Burgers2}, \cite{Cole}, \cite{Enflo} and \cite{Hopf}.

The Riccati equations have played an important role in explicit solutions
for Fisher and Burgers equations (see \cite{Feng}, \cite{Kudray} \ and
references therein). In this paper, in order obtain the main results, we use
a fundamental approach consisting of the use of similarity transformations
and the solutions of Riccati-Ermakov systems with several parameters for the
diffusion case \cite{Suazo:Sus:Ve} and \cite{Suazo:Sus:Ve2}. Of course,
similarity transformations have been extensively applied thanks to Lie
groups and Lie algebras \cite{Bluman}, \cite{Bluman2}, \cite{bluman3} and 
\cite{Olver}. The consideration of parameters in this work is inspired by
the work of E. Marhic, who, in 1978 \cite{Marhic78} introduced (probably for
the first time) a one-parameter $\left\{ \alpha (0)\right\} $ family of
solutions for the linear Schrödinger equation of the one-dimensional
harmonic oscillator. The solutions presented by E. Marhic constituted a
generalization of the original Schrödinger wave packet with oscillating
width. Ermakov systems (for the dispersive case) with solutions containing
parameters \cite{Lan:Lop:Sus} have been used successfully to construct
solutions for the generalized harmonic oscillator with a hidden symmetry 
\cite{Lo:Su:VeSy}, \cite{Lop:Sus:VegaGroup}, and they have also been used to
present Galilei transformation, pseudoconformal transformation and others in
a unified manner, see \cite{Lop:Sus:VegaGroup}. More recently they have been
used in \cite{Mah:Su:Sus} to show spiral and breathing solutions and
solutions with bending for the paraxial wave equation. One of the main
results of this paper is to use similar techniques to provide solutions with
parameters providing a control on the dynamics of solutions (see Figures 1,
2 and 3) for reaction diffusion equations of the form (\ref{GNLH}) and (\ref%
{GBE0}). To this end, it is necessary to establish the conditions on the
coefficients of the differential equations to satisfy Riccati-Ermakov
systems for the diffusion case, which has different solutions compared to
the dispersive case. Once the transformations are established for standard
models such as the Fisher equation or the KPP-equation, explicit global
solutions proposed by Clarkson \cite{Clarkson93} can be used; these
solutions include Jacobi elliptic functions and exponential and rational
functions, see Table 1 for an extended list of examples. Transformations to
the Burgers equation will allow us to produce special exact solutions such
as rational, triangular wave and N-wave type solutions. \ These results
would help to test numerical methods in the study of numerical solutions and
dynamics of singularities. A similar study for nonlinear Schrödinger
equations with variable coefficients can be found in \cite{Escorcia}.

Finite time blow-up for the nonlinear heat equation \cite%
{FilippasHerrero2000}-\cite{HerreroVela1993}, \cite{Kamin85}, \cite{Matano78}%
, \cite{Matos99unfocused}\ and pole dynamics for the Burgers equation have
been studied extensively \cite{calogero}, \cite{chonodnovski}, \cite%
{Deconick:kimura}. We will provide solutions with singularities, which
should motivate further research in the dynamics of singularities.

In general, the variable coefficients Burgers equation (vcBE) is not
integrable, and there are not many exactly solvable models known for vcBE 
\cite{Rosa}-\cite{Scott}. For an interesting application on nonlinear
magnetosonic waves propagating perpendicular to a magnetic neutral sheet,
see \cite{Sakai}. Examples of exact solutions include BE with time-dependent
forcing \cite{Orlo} and with elastic forcing terms \cite{Moreau} and \cite%
{Eule:Friedrich}. Previous methods to obtain explicit solutions include
Green's functions \cite{Zola}, transformations \cite{Sopho} and Cole
transformations \cite{Xu}. In this work we generalize the transformation
presented in \cite{Eule:Friedrich} where Langevin equations and the Hill
equation were used to express the transformation. Instead, we will use what
we have called the Riccati system; further, our solutions will show \textsl{%
multiparameters. }Table 2 shows several examples of families with explicit
solutions and mutiparameters.

This paper is organized in the following manner: In Section 2, we present
Lemma 1 where conditions on the coefficients are established for equation %
\ref{GNLH} to be transformed into the standard Fisher-KPP equation through a
similarity transformation. Further, conditions are given to obtain solutions
with singularities or to avoid them. Examples of equations constructed using
Ermakov systems with solutions without singularities are also explained. \
In Section 3, an alternative approach to solve Riccati systems explicitly is
presented. This approach allows us to present abundant families generalizing
Fisher-KPP equation (see Table 1 with $a=0$) presenting explicit solutions.
Table 2 presents examples of equations with singularities. Conditions for
the existence of solutions are presented. In Section 4, we will study
explicit solutions with multiparameters for the\ variable coefficient
Burgers equation \ref{GBE0}. Several examples with explicit solutions are
presented in Tables 3 and 4. As an applications of the multiparamters
approach we present a new symmetry for burgers equation. We also provide an
appendix where we recall solutions of Riccati-Ermakov systems (for the
diffusion case) previously published \cite{Suazo:Sus:Ve}. Finally, we
provide a supplementary file where our solutions are verified.

\section{Riccati-Ermakov System and Similarity Transformation for Variable
Coefficient Reaction-Diffusion Equations}

In Lemma 1, we show \ref{GNLH} can be transformed into the standard
Fisher-KPP equation through a similarity transformation and as an
application of the multiparameter solution for Riccati-Ermakov systems from
the appendix. Conditions are given to obtain solutions with singularities or
to avoid them. Examples of equations constructed using Ermakov systems with
solutions without singularities are also explained.

\begin{lem}
Let $a(t)>0,$ $\forall t>0$ and $b,c,d,f,g$ suitable time dependent
functions. The variable coefficient nonlinear reaction-diffusion equation 
\begin{align}
\frac{\partial u}{\partial t}=a& \left( t\right) \frac{\partial ^{2}u}{%
\partial x^{2}}-\left( g\left( t\right) -c\left( t\right) x\right) \frac{%
\partial u}{\partial x}+  \notag  \label{GNLH} \\
& \left( d\left( t\right) +L(t)+M(t)x-B(t)x^{2}+h(x,t)u^{p}\right) u 
\tag{GNLH}
\end{align}%
can be reduced to the Fisher-KPP equation ($p>0,$ $x\in 
\mathbb{R}
,$ $t>0$ and $r_{0}$ and $h_{0}$ real constants) 
\begin{equation}
\frac{\partial v}{\partial \tau }=\frac{\partial ^{2}v}{\partial \xi ^{2}}+v(%
{r}_{0}+{h}_{0}v^{p})  \label{ecuacionmodelo}
\end{equation}%
through the similarity transformation 
\begin{align}
u\left( x,t\right) & =\frac{1}{\sqrt{\mu \left( t\right) }}e^{\alpha \left(
t\right) x^{2}+\delta \left( t\right) x+\kappa \left( t\right) }v\left( \xi
,\tau \right) ,  \label{substitution1} \\
\xi & =\beta \left( t\right) x+\varepsilon \left( t\right) ,\quad \tau
=\gamma \left( t\right) ,  \notag
\end{align}%
where\newline
1. $\mu $ satisfies the Ermakov equation (with $c_{0}\in \left\{ 0,1\right\} 
$) 
\begin{equation}
\mu ^{\prime \prime }-\left( \frac{a^{\prime }}{a}+2c-4d\right) \mu ^{\prime
}-4\left( ab+cd+d^{2}-\frac{d}{2}\left( \frac{a^{\prime }}{a}-\frac{%
d^{\prime }}{d}\right) \right) \mu =c_{0}\frac{4a^{2}\exp \left(
4\int_{0}^{t}\left( c\left( s\right) -2d\left( s\right) \right) \ ds\right) 
}{\mu ^{3}},  \label{Ermakov-P}
\end{equation}%
2. $\alpha (t)$, $\beta (t)$, $\gamma (t)$, $\delta (t)$, $\varepsilon (t)$
and $\kappa (t)$ satisfy the Riccati-Ermarkov's system (\ref{Ermakov01})-(%
\ref{Ermakov06}),\newline
3. The following balance between the coefficients holds 
\begin{align}
h(x,t)& =h_{0}a(t)\beta ^{2}(t)\mu ^{p/2}(t)e^{-p(\alpha (t)x^{2}+\delta
(t)x+\kappa (t))},  \label{Conditions 1} \\
B(t)& =b(t)-c_{0}a(t)\beta ^{4}(t), \\
M(t)& =f(t)+2c_{0}a(t)\beta ^{3}(t)\varepsilon (t), \\
L(t)& =a(t)\beta ^{2}(t)(c_{0}\varepsilon ^{2}(t)+r_{0}).
\label{Conditions 4}
\end{align}
\end{lem}

\begin{proof}[\textbf{Proof.}]
To start the proof of the Lemma, let's define $S(x,t)=\alpha (t)x^{2}+\delta
(t)x+k(t).$ Then 
\begin{equation*}
S_{t}(x,t)=\alpha ^{\prime }(t)x^{2}+\delta ^{\prime }x+\kappa ^{\prime
}(t),\quad \quad S_{x}(x,t)=2\alpha (t)x+\delta (t),\quad \quad
S_{xx}(x,t)=2x.
\end{equation*}%
\noindent It is easy to check from (\ref{substitution1}) that 
\begin{equation*}
\dfrac{\partial u}{\partial t}=\dfrac{e^{S(x,t)}}{\sqrt{\mu (t)}}\left[
\alpha ^{\prime }(t)x^{2}v+\delta ^{^{\prime }}(t)xv+k^{\prime }(t)v-\dfrac{%
\mu ^{\prime }(t)}{2\mu (t)}v+v_{\xi }\beta ^{\prime }(t)x+v_{\xi
}\varepsilon ^{\prime }(t)+v_{\tau }\gamma ^{\prime }(t)\right] ,
\end{equation*}%
\begin{equation*}
\dfrac{\partial u(x,t)}{\partial x}=\dfrac{e^{S(x,t)}}{\sqrt{\mu (t)}}\left[
(2\alpha (t)x+\delta (t))v(\xi ,\tau )+v_{\xi }(\xi ,\tau )\beta (t)\right] ,
\end{equation*}%
\noindent

and 
\begin{equation*}
\dfrac{\partial ^{2}u(x,t)}{\partial x^{2}}=\dfrac{e^{S(x,t)}}{\sqrt{\mu (t)}%
}\left[ (2\alpha (t)x+\delta (t))^{2}v+2\beta (t)S_{x}v_{\xi }+2\alpha
(t)v+v_{\xi \xi }\beta ^{2}(t)\right] .
\end{equation*}%
Replacing on the right side of the equation (\ref{GNLH}) and grouping we
obtain

\begin{eqnarray}
\dfrac{d\alpha (t)}{dt} &=&-b(t)+2c(t)\alpha (t)+4a(t)\alpha
^{2}(t)+c_{0}\beta ^{4}(t),  \label{Ermakov01} \\
\dfrac{d\beta (t)}{dt} &=&(c(t)+4a(t)\alpha (t))\beta (t), \\
\dfrac{d\gamma (t)}{dt} &=&a(t)\beta ^{2}(t),  \label{gamma} \\
\dfrac{d\delta (t)}{dt} &=&(c(t)+4a(t)\alpha (t))\delta (t)+f(t)-2\alpha
(t)g(t)+2c_{0}a(t)\beta ^{3}\varepsilon (t), \\
\dfrac{d\varepsilon (t)}{dt} &=&-(g(t)-2a(t)\delta (t))\beta (t), \\
\dfrac{d\kappa (t)}{dt} &=&-g(t)\delta (t)+a(t)\delta
^{2}(t)+c_{0}a(t)\varepsilon ^{2}(t)\beta ^{2}  \label{Ermakov06}
\end{eqnarray}%
\noindent with $c_{0}\in \left\{ 0,1\right\} $. This system will be refered
\ as the Riccati-Ermakov system. We will also use the standard substitution 
\begin{equation}
\alpha (t)=\dfrac{1}{4a(t)}\dfrac{\mu ^{\prime }(t)}{\mu (t)}-\dfrac{d(t)}{%
2a(t)}.
\end{equation}

Therefore, (\ref{GNLH}) becomes 
\begin{align*}
a(t)& \frac{\partial ^{2}u}{\partial x^{2}}-\left( g\left( t\right) -c\left(
t\right) x\right) \frac{\partial u}{\partial x}+\left( d\left( t\right)
+L(t)+M(t)x-B(t)x^{2}+h(x,t)|u|^{p}\right) u \\
=& \left( 4a(t)\alpha ^{2}(t)+2c(t)\alpha (t)-b(t)+c_{0}a(t)\beta
^{4}(t)\right) x^{2}v \\
& +\left( 4a(t)\alpha (t)\delta (t)-2g(t)\alpha (t)+c(t)\delta
(t)+f(t)+2c_{0}a(t)\beta ^{3}(t)\varepsilon (t)\right) xv \\
& +\left( a(t)\delta ^{2}(t)+4a(t)\beta (t)\alpha (t)-g(t)\delta
(t)+c(t)\beta (t)+c_{0}a(t)\beta ^{2}(t)\varepsilon ^{2}(t)+2a(t)\alpha
(t)+d(t)\right) v \\
& +\left( 2a(t)\beta (t)\delta (t)-g(t)\beta (t)\right) v_{\xi }+a(t)\beta
^{2}v_{\xi \xi }+r_{0}a(t)\beta ^{2}(t)v+h(x,t)e^{pS(x,t)}\mu
(t)^{-p/2}|v|^{p}v,
\end{align*}%
\noindent since by hypothesis 2. of the Lemma the functions $\alpha $, $%
\beta $, $\gamma $, $\delta $, $\varepsilon $, $\mu $, and $\kappa $ satisfy
the Riccati-Ermakov system (\ref{Ermakov01})- (\ref{Ermakov06}), and
function $v(\xi ,\tau )$ is the solution of the equation (\ref%
{ecuacionmodelo}). Therefore, using the balance of the coefficients (\ref%
{Conditions 1})- (\ref{Conditions 4}) we have 
\begin{align*}
a\left( t\right) \frac{\partial ^{2}u}{\partial x^{2}}& -\left( g\left(
t\right) -c\left( t\right) x\right) \frac{\partial u}{\partial x}+\left(
d\left( t\right) +L(t)+M(t)x-B(t)x^{2}+h(x,t)|u|^{p}\right) u \\
& =\alpha ^{\prime }(t)x^{2}v+\delta ^{\prime }xv+\kappa ^{\prime }(t)v-%
\dfrac{\mu ^{\prime }(t)}{2\mu (t)}v+\beta ^{\prime }xv_{\xi }+\varepsilon
^{\prime }(t)v_{\xi }+\gamma ^{\prime }(t)v_{\tau } \\
& =\dfrac{\partial u(x,t)}{\partial t},
\end{align*}%
\noindent as was claimed.
\end{proof}

The following Corollary shows that we can have solutions with singularities.

\begin{corollary}
Let's assume the conditions of Lemma 1. If the functions $\alpha (t)$, $%
\beta (t)$, $\gamma (t)$, $\delta (t)$, $\varepsilon (t)$ and $\kappa (t)$
satisfy the Riccati system (\ref{Ermakov01})-(\ref{Ermakov06}) (with $%
c_{0}=0 $) with solution (\ref{MKernel})-(\ref{kappa0}), (with $\mu \left(
0\right) , $ $\beta (0)\neq 0$), then there exists an interval $I$ of time
such that if $-\alpha \left( 0\right) \in \gamma _{0}(I),$ then (\ref{GNLH})
presents a solution with singularity at $T^{\ast }=\gamma _{0}^{-1}(-\alpha
\left( 0\right) )\in I$ of the form (\ref{substitution1}).
\end{corollary}

\begin{proof}
We use the remarkable solution of the equation (\ref{Ermakov-P}) (with $%
c_{0}=0$)\ given explicitly by $\mu \left( t\right) =-2\mu \left( 0\right)
\mu _{0}\left( t\right) \left( \alpha \left( 0\right) +\gamma _{0}\left(
t\right) \right) ,$ where $\mu _{0}$ and $\mu _{1}$ are linear independent
solutions with initial conditions 
\begin{equation}
\mu _{0}\left( 0\right) =0,\quad \mu _{0}^{\prime }\left( 0\right) =2a\left(
0\right) \neq 0\qquad \mu _{1}\left( 0\right) \neq 0,\quad \mu _{1}^{\prime
}\left( 0\right) =0.
\end{equation}

Since there exists an interval $J$ of time with $\mu _{0}\left( t\right)
\neq 0$ for all $t\in J,$ and $\mu _{0}\left( t\right) $ and $\mu _{1}\left(
t\right) $ have been chosen to be linearly independent on an interval, let's
say $J^{\prime },$ we observe that for $t\in J\cap J^{\prime }\equiv I,$ we
get%
\begin{equation*}
\gamma _{0}^{\prime }(t)=\frac{W[\mu _{0}\left( t\right) ,\mu _{1}\left(
t\right) ]}{2\mu _{0}^{2}(t)}\neq 0,
\end{equation*}%
and therefore, from the general expression for $\mu $ the solution (\ref%
{substitution1}) will have a singularitie at $T^{\ast }=\gamma
_{0}^{-1}(-\alpha \left( 0\right) )\in I.$
\end{proof}

In order to construct abundant families using Lemma 1 and Corollary 1 with
explicit solutions. We will use the following explicit solutions for
standard models, taken from \cite{Clarkson93}\textbf{: }

\begin{itemize}
\item Fisher's equation, introduced in 1937, \cite{Fisher37}: it describes
the wave propagation of an advantageous gene in a population 
\begin{equation*}
u_{\tau }=u_{\xi \xi }+u(1-u);
\end{equation*}

an explicit solution is given by \cite{Abl:Zap}%
\begin{equation}
u_{1}(x,t)=\frac{1}{\left[ 1+(\sqrt{2}-1)\exp \left\{ \frac{x}{\sqrt{6}}-%
\frac{5}{6}t\right\} \right] ^{2}}.  \label{u1}
\end{equation}

\item Non-linear heat equation with absorption \cite{Brezis86}: it describes
a diffusion process that competes with an absorption process: 
\begin{equation*}
u_{\tau }=u_{\xi \xi }+Ku^{p+1}.
\end{equation*}%
An explicit solution for the case $K=-1$ and $p=2$ is%
\begin{equation}
u_{2}(x,t)=\frac{\sqrt{2}(2x+k_{1})}{x^{2}+k_{1}x+6t+k_{2}}.  \label{u2}
\end{equation}%
and an explicit solution for the case $K=1$ and $p=3$ is%
\begin{equation}
u_{3}(x,t)=\frac{1}{2}\sqrt{2}(x+k_{1})\sd\left( \frac{1}{2}x^{2}+k_{1}x+3t;%
\frac{1}{2}\sqrt{2}\right) .  \label{u3}
\end{equation}

\item Newell-Whitehead equation 1969 (${r}_{0}=1$ and ${h}_{0}=-1$): it
describes the Rayleigh-Benard convection%
\begin{equation*}
u_{t}=u_{xx}+u({r}_{0}+{h}_{0}u^{p}).
\end{equation*}%
\bigskip An explicit solution for the case ${r}_{0}=-1,$ $h_{0}=-1$ and $p=2$
is%
\begin{equation}
u_{4}(x,t)=\frac{k_{2}\sen(\frac{1}{2}\sqrt{2}x)}{k_{1}\exp \left( \frac{3}{2%
}t\right) +k_{2}\cos \left( \frac{1}{2}\sqrt{2}x\right) },  \label{U4}
\end{equation}%
an explicit solution for the case ${r}_{0}=-2,$ $h_{0}=1$ and $p=2$ is%
\begin{equation}
u_{5}(x,t)=\frac{1}{2}\sqrt{2}k_{1}\sen(x+k_{2})\exp (-3t)\sd(k_{1}\cos
(x+k_{2})\exp (-3t);\frac{1}{2}\sqrt{2})\newline
\label{U5}
\end{equation}%
and an explicit solution for the case ${r}_{0}=2,$ $h_{0}=1$ and $p=2$ is%
\begin{equation}
u_{6}(x,t)=\frac{1}{2}\sqrt{2}k_{1}\senh(x+k_{2})\exp (3t)\sd(k_{1}\cosh
(x+k_{2})\exp (3t);\frac{1}{2}\sqrt{2})\newline
.  \label{u6}
\end{equation}
\end{itemize}

\begin{figure}[tbp]
\begin{center}
{\small 
\subfigure[Solution (\ref{solermakov3}) to equation (\ref{Ermarkov3}), with $k_{1}=1,$ $k_{2}=1,$
$\gamma(0)=0,$ $\mu(0)=1.$]{\includegraphics[scale=0.39]{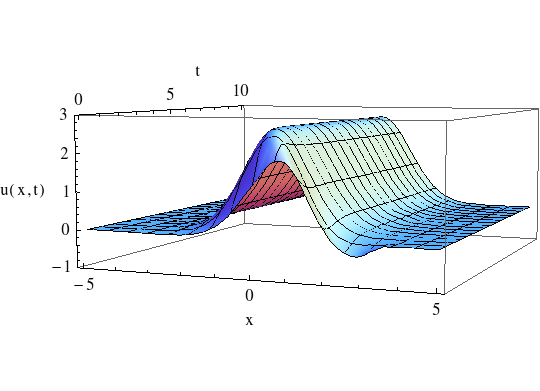}} 
\subfigure[Solution (\ref{solermakov3}) to equation (\ref{Ermarkov3}), with $k_{1}=1,$ $k_{2}=1,$ $\gamma(0)=1,$
$\mu(0)=1.$]{\includegraphics[scale=0.39]{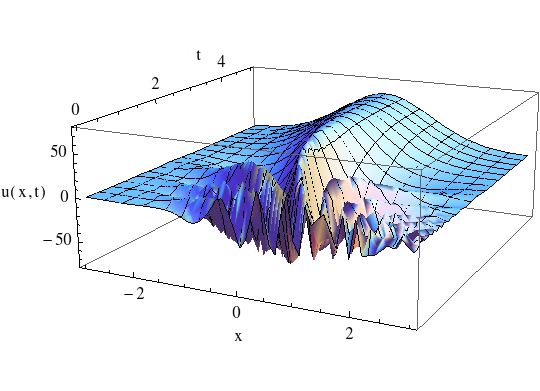}} 
\subfigure[Solution (\ref{solermakov3}) to equation (\ref{Ermarkov3}), with $k_{1}=1,$ $k_{2}=1,$ $\gamma(0)=1,$
$\mu(0)=5.$]{\includegraphics[scale=0.39]{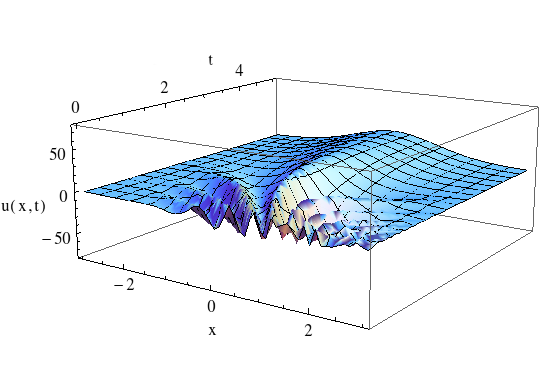}} 
\subfigure[Solution (\ref{solermakov3}) to equation (\ref{Ermarkov3}), with $k_{1}=2,$
$k_{2}=\foreignlanguage{english}{-0.5},$ $\gamma(0)=1,$
$\mu(0)=5.$]{\includegraphics[scale=0.39]{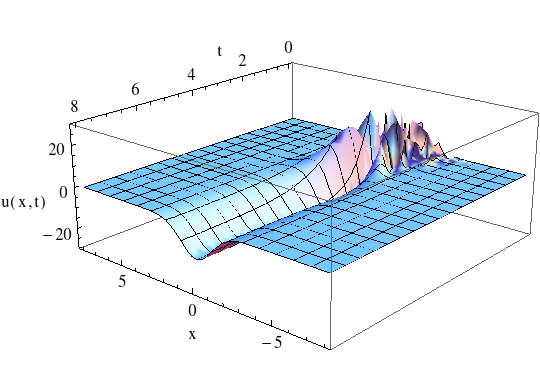}} }
\end{center}
\caption{Jacobi elliptic solutions for reaction diffusion equation (\protect
\ref{Ermarkov3}), with $-5<x<5,$ $0<t<10$. }
\end{figure}

Using Lemma 1, we will give examples of equations with multiparameter
solutions. These toy examples show the control of the dynamics of the
solutions as an application of our multiparameter approach. We have prepared
a Mathematica file as supplemental material for this Section. Also, all the
formulas from the appendix have been verified previously in \cite{Ko:Su;su}
and \cite{Suazo:Sus:Ve2}.

\subsection{Exponential-type solutions for a nonautonomous
reaction-diffusion model}

\label{18} Let $p=2,$ $\ $and $r_{0}=h_{0}=-1$ in (\ref{ecuacionmodelo})$.$
The following nonautonomous nolinear equation 
\begin{equation}
u_{t}=u_{xx}+xu_{x}+\left( \dfrac{2e^{-2t}-e^{-4t}x^{2}+4}{4e^{-2t}+4}%
\right) u-\dfrac{\mu (0)e^{-2t}}{\sqrt{8(e^{-2t}+1)}}\exp \left\{ \dfrac{%
(e^{-2t}+2)x^{2}}{2(1+e^{-2t})}\right\} u^{3}  \label{Ermarkov1}
\end{equation}%
can be transformed to a KPP-type equation, given by 
\begin{equation}
v_{\tau }=v_{\xi \xi }-v(1+v^{2})  \label{Fisher4}
\end{equation}%
with the help of the transformation 
\begin{equation*}
u(x,t)=\dfrac{e^{\alpha (t)x^{2}+\delta (t)x+\kappa (t)}}{\sqrt{\mu (t)}}%
v(\xi ,\tau ),
\end{equation*}%
where $v(\xi ,\tau )$ is a solution of the equation (\ref{Fisher4}), in this
case given by (\ref{U4}) and the parametric functions are given by: 
\begin{align}
\mu (t)& =\mu (0)\sqrt{\dfrac{e^{-2t}+1}{2}},  \label{ermarkov333} \\
\alpha (t)& =-\dfrac{e^{-2t}+2}{4(1+e^{-2t})}, \\
\gamma (t)& =\gamma (0)-\dfrac{1}{4}\ln \left[ \dfrac{e^{-2t}+1}{2}\right] ,
\\
\beta (t)& =\dfrac{e^{-t}}{\sqrt{2(e^{-2t}+1)}}, \\
\delta (t)& =\kappa (t)=\varepsilon (t)=0.  \label{ermakov3}
\end{align}%
Therefore, by Lemma 1 the differential equation (\ref{Ermarkov1}) allows the
following family of solutions with parameters $\mu (0),$ $\gamma (0),$ $%
k_{1},$ $k_{2}:$ 
\begin{equation}
u(x,t)=\dfrac{1}{\sqrt{\mu (0)\sqrt{\frac{e^{-2t}+1}{2}}}}\exp \left\{ -%
\dfrac{x^{2}(e^{-2t}+2)}{4(1+e^{-2t})}\right\} \dfrac{\sen(\sqrt{2}x\beta
(t)/2)k_{2}}{k_{1}\exp \left\{ 3\gamma (t)/2\right\} +k_{2}\cos (\sqrt{2}%
x\beta (t)/2)}.  \label{ermarkov1}
\end{equation}

\subsection{ Parameter variation in the equation (\protect\ref{Ermarkov1}):
Jacobi elliptic-type solution}

\noindent If we make a change in the parameter in (\ref{Ermarkov1}), $%
r_{0}=h_{0}=1,$ by Lemma 1 the equation 
\begin{equation}
u_{t}=u_{xx}+xu_{x}+\left( \dfrac{6e^{-2t}-e^{-4t}x^{2}+4}{4e^{-2t}+4}%
\right) u+\dfrac{\mu (0)e^{-2t}}{\sqrt{8(e^{-2t}+1)}}\exp \left\{ \dfrac{%
(e^{-2t}+2)x^{2}}{2(1+e^{-2t})}\right\} u^{3}  \label{Ermarkov3}
\end{equation}%
\noindent can be reduced to a KPP type model, given by 
\begin{equation}
v_{\tau }=v_{\xi \xi }+v(1+v^{2}).  \label{KPP25}
\end{equation}

\noindent Note that changing parameters $r_{0}$ and $h_{0}$ only affect the
signs in the diferential equation. However, equation (\ref{KPP25}) has
solution (\ref{u6}). \noindent Our differential model Equation (\ref%
{Ermarkov3}) has a solution with parameters $\mu (0),$ $\gamma (0),$ $k_{1}$
y $k_{2}$: 
\begin{align}
u(x,t)=& \dfrac{1}{\sqrt{\mu (0)\sqrt{\frac{e^{-2t}+1}{2}}}}\exp \left\{ -%
\dfrac{x^{2}(e^{-2t}+2)}{4(1+e^{-2t})}\right\} \times  \notag
\label{solermakov3} \\
& \dfrac{1}{2}\sqrt{2}k_{1}\senh(\xi +k_{2})e^{3t}\sd(k_{1}\cosh (\xi
+k_{2})e^{3t};\frac{\sqrt{2}}{2}).
\end{align}%
\noindent A graphic representation of the solution is presented in Figure 1;
in this case the parameter $\gamma (0)$ control the behavior of the function
in time, while $\mu (0)$ determines the amplitude of the solution. 
\begin{table}[tbp]
\caption{Families of generalized Fisher-KPP equations with explicit
solutions of the form (\protect\ref{sustitucionparticular}) through the
alternative Riccati system method approach, $v_{j}$ to be chosen from (%
\protect\ref{u1})- (\protect\ref{u6}).}{\small {\fontfamily{cmr10}%
\selectfont{
\begin{tabular}{|c|c|c|}
\hline
\textbf{\#} & \textbf{Nonlinear Reaction-Diffusion Equation} & \textbf{Solution} \\ \hline
1 & 
\begin{tabular}{c}
$u_{t}=u_{xx}+at^{m}xu_{x}+\left( r_{0}-\frac{1}{4}\left(
amt^{m-1}-a^{2}t^{2m}\right) x^{2}\right) u$ \\ 
$+h_{0}\left[ \exp \{\frac{t^{m+1}}{m+1}\}\right] ^{ap/2}\exp \{p\frac{a}{4}t^{m}x^{2}\}u^{p+1}$\end{tabular}
& $\frac{1}{\left[ \exp \{\frac{t^{m+1}}{m+1}\}\right] ^{a/2}}\exp \{-\frac{a}{4}t^{m}x^{2}\}v_{j}$ \\ \hline
2 & 
\begin{tabular}{c}
$u_{t}=u_{xx}-axu_{x}+\left( r_{0}+\frac{1}{4}a^{2}x^{2}\right) u$ \\ 
$+h_{0}\left[ \exp \{-t\}\right] ^{ap/2}\exp \{-p\frac{a}{4}x^{2}\}u^{p+1}$\end{tabular}
& $\frac{1}{\left[ \exp \{-t\}\right] ^{a/2}}\exp \{\frac{a}{4}x^{2}\}v_{j}$
\\ \hline
3 & 
\begin{tabular}{c}
$u_{t}=u_{xx}+ae^{\lambda t}xu_{x}+\left( r_{0}-\frac{1}{4}(a\lambda
e^{\lambda t}-a^{2}e^{2\lambda t})x^{2}\right) u$ \\ 
$+h_{0}\exp \{\left[ \frac{a}{\lambda }\left( e^{\lambda t}-1\right) \right]
\}^{p/2}\exp \{p\frac{a}{4}e^{\lambda t}x^{2}\}u^{p+1}$\end{tabular}
& $\frac{1}{\exp \{\frac{a}{\lambda }\left( e^{\lambda t}-1\right) \}^{1/2}}\exp \{-\frac{a}{4}e^{\lambda t}x^{2}\}v_{j}$ \\ \hline
4 & 
\begin{tabular}{c}
$u_{t}=u_{xx}+a\ln (t)xu_{x}+\left( r_{0}-\frac{1}{4}\left( \frac{a}{t}-a^{2}(\ln t)^{2}\right) x^{2}\right) u$ \\ 
$+h_{0}\left[ \left( te^{-1}\right) ^{at}\right] ^{p/2}t^{\frac{pa}{4}x^{2}}u^{p+1}$\end{tabular}
& $\frac{1}{\left( te^{-1}\right) ^{bt/2}}\exp \{-\frac{1}{4}b\ln
tx^{2}\}v_{j}$ \\ \hline
5 & 
\begin{tabular}{c}
${\small u}_{t}{\small =u}_{xx}{\small +ae}^{\lambda t^{2}}{\small xu}_{x}{\small +}\left( r_{0}-\left( \frac{a}{2}\lambda te^{\lambda t^{2}}-\frac{a^{2}}{4}e^{2\lambda t^{2}}\right) x^{2}\right) {\small u}$ \\ 
${\small +h}_{0}\left[ \exp \{\frac{a}{2}\sqrt{\frac{\pi }{\lambda }}\erf i\left( \sqrt{\lambda }t\right) \}\right] ^{p/2}\exp {\small \{p}\frac{a}{4}{\small e}^{\lambda t^{2}}{\small x}^{2}{\small \}u}^{p+1}$\end{tabular}
& $\frac{1}{\left[ \exp \{\frac{a}{2}\sqrt{\frac{\pi }{\lambda }}\erf i\left( \sqrt{\lambda }t\right) \}\right] ^{1/2}}\exp {\small \{-}\frac{a}{4}{\small e}^{\lambda t^{2}}{\small x}^{2}{\small \}v}_{j}$ \\ \hline
6 & 
\begin{tabular}{c}
${\small u}_{t}{\small =u}_{xx}{\small +a}\tanh {\small (\lambda t)xu}_{x}$
\\ 
${\small +}\left( r_{0}+\frac{a}{4}\left( \tanh ^{2}(\lambda t)(a+\lambda
)-\lambda \right) x^{2}\right) {\small u}$ \\ 
${\small +h}_{0}{\small |}\cosh {\small (\lambda t)|}^{ap/(2\lambda )}\exp 
{\small \{p}\frac{a}{4}\tanh {\small (\lambda t)x}^{2}{\small \}u}^{p+1}$\end{tabular}
& $\frac{1}{|\cosh (\lambda t)|^{a/(2\lambda )}}\exp {\small \{-}\frac{a}{4}\tanh {\small (\lambda t)x}^{2}{\small \}v}_{j}$ \\ \hline
7 & 
\begin{tabular}{c}
${\small u}_{t}{\small =u}_{xx}{\small +a}\coth {\small (\lambda t)xu}_{x}{\small +}r_{0}+$ \\ 
$\left( r_{0}+\frac{a}{4}\left( \coth ^{2}(\lambda t)(a+\lambda )-\lambda
\right) x^{2}\right) {\small u}$ \\ 
${\small +h}_{0}{\small |}\sinh {\small (\lambda t)|}^{ap/(2\lambda )}\exp 
{\small \{p}\frac{a}{4}\coth {\small (\lambda t)x}^{2}{\small \}u}^{p+1}$\end{tabular}
& $\frac{1}{|\sinh (\lambda t)|^{a/(2\lambda )}}\exp {\small \{-}\frac{a}{4}\coth {\small (\lambda t)x}^{2}{\small \}v}_{j}$ \\ \hline
8 & 
\begin{tabular}{c}
${\small u}_{t}{\small =u}_{xx}{\small +a}\cosh {\small (\lambda t)xu}_{x}$
\\ 
${\small +}\left( r_{0}+\frac{a}{4}\left( a-\lambda \sinh (\lambda t)+a\sinh
^{2}(\lambda t)\right) x^{2}\right) {\small u}$ \\ 
${\small +h}_{0}\exp {\small \{}\frac{a}{\lambda }\sinh {\small (\lambda t)\}}^{p/2}\exp {\small \{p}\frac{a}{4}\cosh {\small (\lambda t)x}^{2}{\small \}u}^{p+1}$\end{tabular}
& $\frac{1}{\left[ \exp \{\frac{a}{\lambda }\sinh (\lambda t)\}\right] ^{1/2}}\exp {\small \{-}\frac{a}{4}\cosh {\small (\lambda t)x}^{2}{\small \}v}_{j}$
\\ \hline
9 & 
\begin{tabular}{c}
${\small u}_{t}{\small =u}_{xx}{\small -a}\cos {\small (\lambda t)xu}_{x}$
\\ 
${\small +}\left( r_{0}+\frac{a}{4}\left( a-\lambda \sin (\lambda t)+a\sin
^{2}(\lambda t)\right) x^{2}\right) {\small u}$ \\ 
${\small +h}_{0}\exp {\small \{-}\frac{a}{\lambda }\sin {\small (\lambda t)\}}^{p/2}\exp {\small \{-p}\frac{a}{4}\cos {\small (\lambda t)x}^{2}{\small \}u}^{p+1}$\end{tabular}
& $\frac{1}{\left[ \exp \{-\frac{a}{\lambda }\sin (\lambda t)\}\right] ^{1/2}}\exp {\small \{}\frac{a}{4}\cos {\small (\lambda t)x}^{2}{\small \}v}_{j}$
\\ \hline
10 & 
\begin{tabular}{c}
${\small u}_{t}{\small =u}_{xx}{\small +a}\sin {\small (\lambda t)xu}_{x}$
\\ 
${\small +}\left( r_{0}+\frac{a}{4}\left( a-\lambda \cos (\lambda t)+a\cos
^{2}(\lambda t)\right) x^{2}\right) {\small u}$ \\ 
${\small +h}_{0}\exp {\small \{-}\frac{a}{\lambda }\cos {\small (\lambda t)\}}^{p/2}\exp {\small \{p}\frac{a}{4}\sin {\small (\lambda t)x}^{2}{\small \}u}^{p+1}$\end{tabular}
& $\frac{1}{\left[ \exp \{-\frac{a}{\lambda }\cos (\lambda t)\}\right] ^{1/2}}\exp {\small \{-}\frac{a}{4}\sin {\small (\lambda t)x}^{2}{\small \}v}_{j}$
\\ \hline
11 & 
\begin{tabular}{c}
${\small u}_{t}{\small =u}_{xx}{\small +a}\tan {\small (\lambda t)xu}_{x}$
\\ 
${\small +}\left( r_{0}+\frac{a}{4}\left( \tan ^{2}(\lambda t)(a+\lambda
)-\lambda \right) x^{2}\right) {\small u}$ \\ 
${\small +h}_{0}{\small |}\cos {\small (\lambda t)|}^{-ap/(2\lambda )}\exp 
{\small \{p}\frac{a}{4}\tan {\small (\lambda t)x}^{2}{\small \}u}^{p+1}$\end{tabular}
& $\frac{1}{|\cos (\lambda t)|^{-a/(2\lambda )}}\exp {\small \{-}\frac{a}{4}\tan {\small (\lambda t)x}^{2}{\small \}v}_{j}$ \\ \hline
12 & 
\begin{tabular}{c}
${\small u}_{t}{\small =u}_{xx}{\small -a}\cot {\small (\lambda t)xu}_{x}$
\\ 
${\small +}\left( r_{0}+\frac{a}{4}\left( \cot ^{2}(\lambda t)(a+\lambda
)-\lambda \right) x^{2}\right) {\small u}$ \\ 
${\small +h}_{0}{\small |}\sin {\small (\lambda t)|}^{-ap/(2\lambda )}\exp 
{\small \{-p}\frac{a}{4}\cot {\small (\lambda t)x}^{2}{\small \}u}^{p+1}$\end{tabular}
& $\frac{1}{|\sin (\lambda t)|^{-a/(2\lambda )}}\exp {\small \{}\frac{a}{4}\cot {\small (\lambda t)x}^{2}{\small \}v}_{j}$ \\ \hline
\end{tabular}\label{tabla}
}}\newline
\ }
\end{table}

\section{Alternative Riccati system method and Similarity Transformation for
Variable Coefficient Reaction Difussion Equations}

One diffculty of applying Lemma 1 is solving the Ermakov system. In this
section, we present an alternative approach to deal with the Riccati system
and see how the dynamics of the solutions change with multiparameters. If we
choose $\beta (t)\equiv 1,$ $\varepsilon (t)\equiv 0,$ $\gamma (t)\equiv
\tau \equiv t$ in the Riccati system, we obtain the following relation: 
\begin{equation*}
\alpha (t)=-\frac{1}{4}c(t),\hspace{0.7cm}\delta (t)=\frac{g(t)}{2}.
\end{equation*}

Further, the Riccati system (\ref{Ermakov01})-(\ref{Ermakov06}) with $%
c_{0}=0 $ is reduced to%
\begin{align}
c^{\prime }(t)& -c^{2}(t)-4(b(t)-c_{0})=0,  \label{sistemaparticular11} \\
\mu ^{\prime }(t)& -c(t)\mu (t)+2d(t)\mu (t)=0, \\
\delta ^{\prime }(t)& -\dfrac{1}{2}c(t)g(t)-f(t)=0, \\
\kappa ^{\prime }(t)& +\dfrac{g^{2}(t)}{4}=0  \label{sistemaparticular12}
\end{align}%
with initial conditions $\delta (0)=\frac{1}{2}g(0),$ $\mu (0)>0$ and $g(0),$
$c(0),$ $\kappa (0),$ constants. The solution of the particular Riccati
system (\ref{sistemaparticular11})-(\ref{sistemaparticular12}) is given
explicitly by 
\begin{align}
c^{^{\prime }}(t)-c^{2}(t)& -4(b(t)-c_{0})=0,  \label{sistemafinal11} \\
\alpha (t)& =-\dfrac{1}{4}c(t), \\
\delta (t)& =\dfrac{1}{2}g(t), \\
\kappa (t)& =-\dfrac{1}{4}\int_{0}^{t}g^{2}(z)dz+\kappa (0), \\
\mu (t)& =\mu (0)\exp \left\{ \int_{0}^{t}(c(z)-2d(z))dz\right\} , \\
g(t)=2\exp \left\{ \int_{0}^{t}c(z)dz\right\} & \int_{0}^{t}f(z)\exp \left\{
-\int_{0}^{z}c(w)dw\right\} dz+g(0),  \label{sistemafinal12}
\end{align}

subject to $\delta (0)=\frac{1}{2}g(0)$ $\mu (0)>0$ and $c(0),$ $\kappa (0),$
arbitrary constants.

\bigskip 
\begin{table}[tbp]
\caption{Examples of equations with singularities of the form (\protect\ref%
{sustitucionparticular}). $v_{i}$ to be chosen from (\protect\ref{u1})- (%
\protect\ref{u6}).}{\small {\fontfamily{cmr10}%
\selectfont{
\begin{tabular}{|l|l|l|}
\hline
\# & \textbf{Nonlinear Reaction-Difusion Equation} & \textbf{Solution} \\ 
\hline
13 & 
\begin{tabular}{c}
$u_{t}=u_{xx}-\frac{1}{t}xu_{x}+r_{0}u$ \\ 
$+h_{0}t^{-p/2}\exp \{-p\frac{1}{4t}x^{2}\}u^{p+1}$\end{tabular}
& $\frac{1}{t^{1/2}}\exp \{\frac{1}{4t}x^{2}\}v_{j}$ \\ \hline
14 & 
\begin{tabular}{c}
$u_{t}=u_{xx}-\coth (t)xu_{x}+\left( r_{0}+\frac{1}{4}x^{2}\right) u$ \\ 
$+h_{0}|\csch(t)|^{p/2}\exp \{-p\frac{1}{4}\coth (t)x^{2}\}u^{p+1}$\end{tabular}
& $\frac{1}{|\csch(t)|^{1/2}}\exp \{\frac{1}{4}\coth (t)x^{2}\}v_{j}$ \\ 
\hline
15 & 
\begin{tabular}{c}
$u_{t}=u_{xx}-\tanh (t)xu_{x}+\left( r_{0}+\frac{1}{4}x^{2}\right) u$ \\ 
$+h_{0}|\sech(t)|^{p/2}\exp \{-p\frac{1}{4}\tanh (t)x^{2}\}u^{p+1}$\end{tabular}
& $\frac{1}{|\sech(t)|^{1/2}}\exp \{\frac{1}{4}\tanh (t)x^{2}\}v_{j}$ \\ 
\hline
16 & 
\begin{tabular}{c}
$u_{t}=u_{xx}-\cot (t)xu_{x}+\left( r_{0}-\frac{1}{4}x^{2}\right) u$ \\ 
$+h_{0}|\csc (t)|^{p/2}\exp \{-\frac{p}{4}\cot (t)x^{2}\}u^{p+1}$\end{tabular}
& $\frac{1}{|\csc (t)|^{1/2}}\exp \{\frac{1}{4}\cot (t)x^{2}\}v_{j}$ \\ 
\hline
17 & 
\begin{tabular}{c}
$u_{t}=u_{xx}+\tan (t)xu_{x}+\left( r_{0}-\frac{1}{4}x^{2}\right) u$ \\ 
$+h_{0}|\sec (t)|^{p/2}\exp \{\frac{p}{4}\tan (t)x^{2}\}u^{p+1}$\end{tabular}
& $\frac{1}{|\sec (t)|^{1/2}}\exp \{-\frac{1}{4}\tan (t)x^{2}\}v_{j}$ \\ 
\hline
\end{tabular}\label{tabla}
}}\newline
\ }
\end{table}

\begin{lem}
\label{resultadoparticular} Let $p>0,$ $x\in 
\mathbb{R}
$ and $t>0.$ If we define 
\begin{equation*}
h(x,t)=h_{0}\mu ^{p/2}(t)\exp \left\{ -p(\alpha (t)x^{2}+\delta (t)x+\kappa
(t))\right\}
\end{equation*}%
and the functions $c(t),$ $g(t),$ $\alpha (t),$ $\kappa (t),$ $\delta (t)$
such that the equations (\ref{sistemafinal11})-(\ref{sistemafinal12}) are
satisfied, then the function 
\begin{equation}
u(x,t)=\dfrac{1}{\sqrt{\mu (t)}}\exp \left\{ \alpha (t)x^{2}+\delta
(t)x+\kappa (t)\right\} v(x,t)  \label{sustitucionparticular}
\end{equation}%
\noindent is a solution of the equation 
\begin{equation}
u_{t}=u_{xx}-(g(t)-c(t)x)u_{x}+(d(t)+r_{0}+f(t)x-(b(t)-c_{0})x^{2})u+h(x,t)u^{p+1}
\label{GNLHparticular}
\end{equation}%
\noindent where $v(x,t)$ is a solution of the equation ($r_{0},$ $h_{0}$
constants) 
\begin{equation}
v_{t}=v_{xx}+v(r_{0}+h_{0}v^{p}).  \label{ecuacionmodeloparticular}
\end{equation}
\end{lem}

The proof is similar to Lemma 1.

\begin{remark}
The solutions presented in Tables 1 and 2 are of the form (\ref%
{sustitucionparticular}) where for convenience of the presentation $\mu ,$ $%
\alpha ,\beta ,\gamma ,\delta ,$ $\varepsilon $ and $\kappa $ satisfies (\ref%
{alpha0})-(\ref{kappa0}), but we would like to emphasize that the most
general multiparameter solutions are given by (\ref{sustitucionparticular})
and $\mu ,$ $\alpha ,\beta ,\gamma ,\delta ,$ $\varepsilon $ and $\kappa $
satisfies (\ref{AKernel})-(\ref{kappa0}).\bigskip
\end{remark}

\subsection{ Jacobi elliptic-type solutions for a nonautonomous
reaction-diffusion model.}

Choosing $c(t)=\frac{\sech^{2}(t/2)}{2+2\tanh (t/2)},$ let $x\in 
\mathbb{R}
,$ $t>0,$ $p=2,$ $r_{0}=-2,$ $c_{0}=h_{0}=1.$ Following Lemma \ref%
{resultadoparticular} the nonlinear nonautonomous equation 
\begin{align}
u_{t}& =u_{xx}+\left( -1+\frac{1}{2}e^{t/2}\tanh \left( \frac{t}{2}\right) x+%
\frac{1}{8}\left( 1+\tanh \left( \frac{t}{2}\right) \right) x^{2}\right) u 
\notag  \label{Ermarkov2} \\
& -\left( 2e^{t/2}-\dfrac{\sech^{2}\left( \frac{t}{2}\right) }{2\left(
1+\tanh \left( \frac{t}{2}\right) \right) }\right) xu_{x}+\exp \left\{
-2xe^{t/2}+2e^{t}-2+2\kappa (0)\right\} \\
& \times \mu (0)\sech\left( \frac{t}{2}\right) \exp \left\{ \frac{-3t}{2}+%
\frac{\sech^{2}\left( \frac{t}{2}\right) }{4(1+\tanh \left( \frac{t}{2}%
\right) )}x^{2}\right\} u^{3}  \notag
\end{align}%
can be translated to the Fisher-Kolmogorov type model 
\begin{equation}
v_{t}=v_{xx}+v(v^{2}-2)  \label{trans1}
\end{equation}%
by the transformation (\ref{sustitucionparticular}) \noindent where $v(x,t)$
is solution of the equation (\ref{trans1}), given by (\ref{U5}), \noindent
and the parametric functions $\alpha (t),$ $\mu (t),$ $\kappa (t),$ $\delta
(t),$ are given by 
\begin{align}
\alpha (t)& =-\dfrac{\sech^{2}(\frac{t}{2})}{8(1+\tanh (\frac{t}{2}))}, \\
\mu (t)& =\mu (0)\sech\left( \frac{t}{2}\right) e^{\frac{-3t}{2}}, \\
\delta (t)& =e^{t/2}, \\
\kappa (t)& =1-e^{t}+\kappa (0).
\end{align}%
Therefore, Lemma \ref{resultadoparticular} ensures that equation (\ref%
{Ermarkov2}) allows the following global solution in terms of exponential
and Jacobi elliptic functions: 
\begin{align}
u(x,t)=& \dfrac{\exp \left\{ -\dfrac{\sech^{2}(\frac{t}{2})}{8(1+\tanh (%
\frac{t}{2}))}x^{2}+e^{\frac{t}{2}}x-e^{t}+1+\kappa (0)\right\} }{\sqrt{\mu
(0)\sech(\frac{t}{2})}e^{\frac{-3t}{4}}}  \label{ermarkov2} \\
& \times \dfrac{\sqrt{2}}{2}k_{1}\sen(x+k_{2})e^{-3t}\sd\left( k_{1}\cosh
(x+k_{2})e^{-3t};\dfrac{1}{2}\sqrt{2}\right) .  \notag
\end{align}%
\begin{figure}[h]
{\small \centering
\subfigure[Solution (\ref{ermarkov2}) with $k_{1}=1,$ $k_{2}=2,$ $\kappa(0)=0,$ $\mu(0)=1.$]{
\includegraphics[scale=0.38]{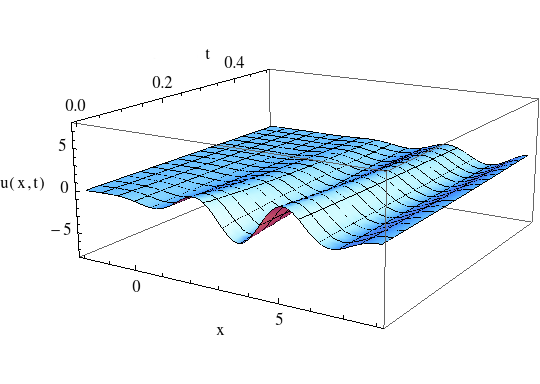}} 
\subfigure[Solution(\ref{ermarkov2}) with $k_{1}=1,$ $k_{2}=2,$ $\kappa(0)=5,$ $\mu(0)=10.$]{
\includegraphics[scale=0.38]{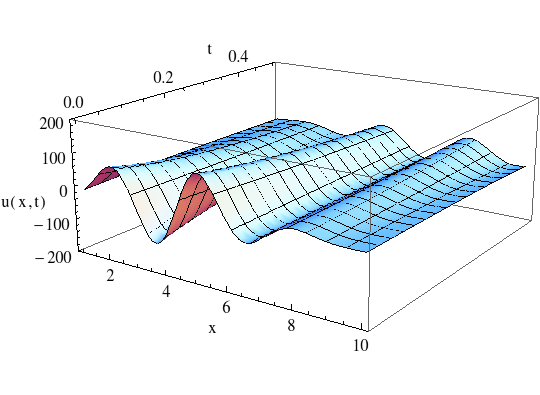} } 
\subfigure[Solution
(\ref{ermarkov2}) with
$k_{1}=3,$ $k_{2}= \foreignlanguage{english}{0.5},$ $\kappa(0)=1,$
$\mu(0)=5.$]{\includegraphics[scale=0.38]{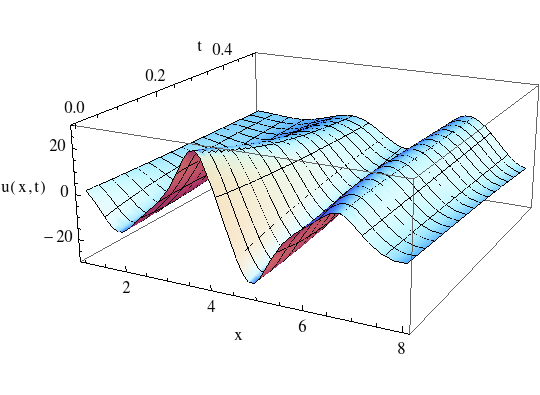}} }
\caption{ Particular solutions for the reaction-diffusion model (\protect\ref%
{Ermarkov2}), with $2<x<3,$ $0<t<\foreignlanguage{english}{0.5}.$}
\label{ermarkovf2}
\end{figure}

\subsection{ Rational solution for a nonautonomous reaction-diffusion model:}

Choosing $c(t)=\csch(t)\sech(t).$ Let $x\in 
\mathbb{R}
,$ $t>0,$ $p=2,$ $r_{0}=c_{0}=0,$ $h_{0}=-1.$ By Lemma \ref%
{resultadoparticular}, it follows that the nonlinear nonautonomous equation 
\begin{align}
u_{t}=& u_{xx}+\csch t\sech txu_{x}+\dfrac{1}{2}\left( 1+\csch t\sech t+\csch%
^{2}tx^{2}\right) u  \label{Riccati1} \\
& -\mu (0)\exp \left\{ \dfrac{1}{2}\csch t\sech tx^{2}-2\kappa (0)-t\right\}
u^{3}  \notag
\end{align}%
can be reduced to nonlinear absorption model given by 
\begin{equation}
v_{t}=v_{xx}-v^{3}  \label{power3}
\end{equation}%
by the transformation (\ref{sustitucionparticular}) \ where $v(x,t)$ is a
solution of the equation (\ref{power3}), given by (\ref{u2}), and the
parametric functions $\alpha (t),$ $\mu (t),$ $\kappa (t),$ $\delta (t)$ are
given by 
\begin{align}
\alpha (t)& =-\dfrac{1}{4}\csch t\sech t, \\
\mu (t)& =\mu (0)e^{-t}, \\
\kappa (t)& =\kappa (0), \\
\delta (t)& =0.
\end{align}

Therefore, Lemma \ref{resultadoparticular} ensures that equation (\ref%
{Riccati1}) allows the following global solution with exponential functions,
with parameters $\mu (0),$ $\kappa (0),$ $k_{1}$ and $k_{2}$ determining the
amplitude of the function 
\begin{equation*}
u(x,t)=(\mu (0))^{-1/2}\exp \left\{ -\dfrac{1}{4}\csch t\sech tx^{2}+\kappa
(0)+\dfrac{t}{2}\right\} \dfrac{\sqrt{2}(2x+k_{1})}{x^{2}+k_{1}x+6t+k_{2}}.
\end{equation*}

\subsection{Jacobi elliptic-type solution for a nonautonomous
reaction-diffusion model}

Choosing $c(t)=\tanh (t).$ Let $x\in 
\mathbb{R}
,$ $t>0,$ and $p=2,$ $r_{0}=c_{0}=0,$ $h_{0}=1.$ From Lemma \ref%
{resultadoparticular} it follows that the following nonautonomous nonlinear
equation 
\begin{align}
u_{t}=& u_{xx}+2\sech tu_{x}+\tanh txu_{x}+\left( 1+2x\tanh t\sech t-\dfrac{1%
}{4}(1-2\tanh ^{2}t)x^{2}\right) u  \notag  \label{Riccati2} \\
& +\mu (0)\exp \left\{ -2\left( -\dfrac{1}{4}\tanh tx^{2}-x\sech t+\kappa
(0)-\tanh t\right) -2t\right\} \cosh tu^{3}
\end{align}%
can be reduced to the absorption model given by 
\begin{equation}
v_{t}=v_{xx}+v^{3}  \label{+power3}
\end{equation}%
by the transformation 
\begin{equation*}
u(x,t)=\dfrac{e^{\alpha (t)x^{2}+\delta (t)x+\kappa (t)}}{\sqrt{\mu (t)}}%
v(x,t),
\end{equation*}%
where $v(x,t)$ is a solution of the equation (\ref{+power3}), given by 
\begin{equation*}
v(x,t)=\dfrac{\sqrt{2}}{2}(x+k_{1})\sd\left( \dfrac{1}{2}x^{2}+k_{1}x+3t;%
\dfrac{1}{2}\sqrt{2}\right) ,
\end{equation*}%
and the parametric functions $\alpha (t),$ $\mu (t),$ $\kappa (t),$ $\delta
(t),$ are given by: 
\begin{align}
\alpha (t)& =-\dfrac{1}{4}\tanh t, \\
\delta (t)& =-\sech t, \\
\kappa (t)& =\kappa (0)-\tanh t, \\
\mu (t)& =\mu (0)\cosh t\exp (-2t).
\end{align}

We obtain from Lemma \ref{resultadoparticular} that equation (\ref{Riccati2}%
) allows the following global solution with exponential and Jacobi elliptic
functions with parameters $\mu (0),$ $\kappa (0),$ which determine the
amplitude: 
\begin{align}
u(x,t)=& \dfrac{\exp \left\{ -\frac{1}{4}\tanh tx^{2}-x\sech t+\kappa
(0)-\tanh t\right\} }{\sqrt{\mu (0)\cosh t\exp (-2t)}}  \label{riccati2} \\
& \times \dfrac{\sqrt{2}}{2}(x+k_{1})\sd\left( \dfrac{1}{2}x^{2}+k_{1}x+3t;%
\dfrac{1}{2}\sqrt{2}\right) .  \notag
\end{align}%
\begin{figure}[h]
{\small \centering
\subfigure[Solution (\ref{riccati2}), with $k_{1}=1,$ $\kappa(0)=- \foreignlanguage{english}{0.5},$ $\mu(0)=20.$]{
\includegraphics[scale=0.37]{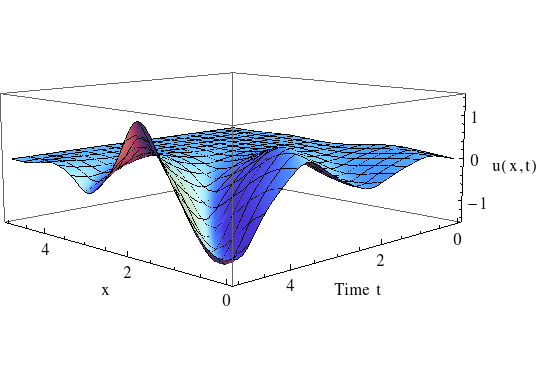}} 
\subfigure[Solution (\ref{riccati2}), with $k_{1}=1,$ $\kappa(0)= \foreignlanguage{english}{0.5},$ $\mu(0)=1.$]{
\includegraphics[scale=0.37]{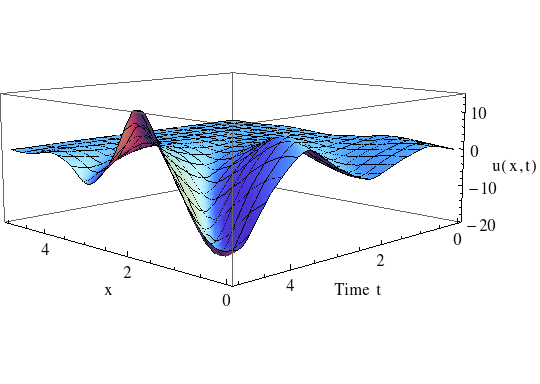} } 
\subfigure[Solution (\ref{riccati2}), with $k_{1}=3,$  $\kappa(0)=\foreignlanguage{english}{0.5},$ $\mu(0)=10.$]{
\includegraphics[scale=0.37]{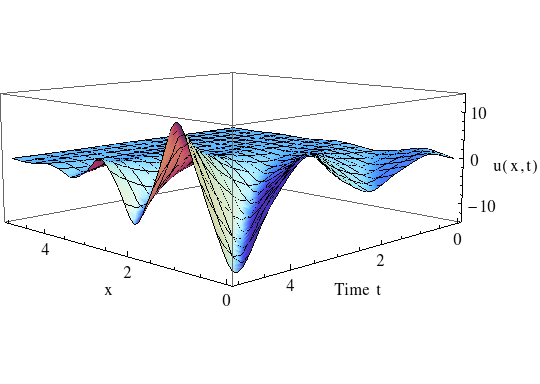}} 
\subfigure[Solution (\ref{riccati2}), with $k_{1}=7,$  $\kappa(0)=\foreignlanguage{english}{0.5},$ $\mu(0)=10.$]{
\includegraphics[scale=0.37]{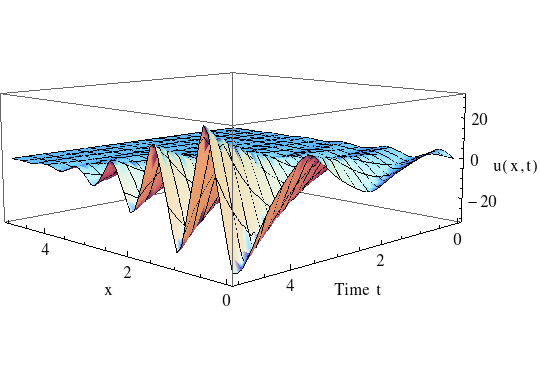}} }
\caption{Particular solutions for the Equation (\protect\ref{Riccati2}), $%
0<x<5,$ $0<t<5.$}
\label{riccati2f}
\end{figure}

\label{Existencia} We recall a result of existence and uniqueness for the
standard nonlinear heat equation (\ref{NLH}) \cite{Kamin85} and the
subsequent analogous existence result for the nonautonomous nonlinear heat
equation (\ref{GNLHparticular}). The space $C_{0}(%
\mathbb{R}
)$ denotes the set of the continuous functions converging to $0$ when $%
|x|\longrightarrow \infty ,$ with the norm $\parallel f(x)\parallel _{\infty
}=\sup_{x}|f(x)|.$ Moreover the space $C([0,\infty ),C_{0}(%
\mathbb{R}
))$ denotes the set of the continuous functions defined in $[0,\infty )$
with values in the space $C_{0}(%
\mathbb{R}
).$

\begin{teo}[\protect\cite{Kamin85}]
\label{existencia} Given $p>0$ consider the heat equation 
\begin{equation}
\dfrac{\partial u}{\partial t}=\dfrac{\partial ^{2}u}{\partial x^{2}}%
-|u|^{p}u,  \label{NLH}
\end{equation}%
$x\in R,$ $t>0.$ If $u_{0}\in C_{0}\left( \mathbb{R}\right) ,$ there exists
a unique global solution of (\ref{NLH}) $u\in C\left( \left[ 0,\infty
\right) ,C_{0}(\mathbb{R})\right) ,$ such that $u(x,t)$ satisfies the
initial condition $u(x,0)=u_{0}(x).$
\end{teo}

The following result proves the existence of an initial value problem
solution associated with the generalized nonautonomous heat equation (\ref%
{GNLHparticular}).

\begin{proposition}
\label{existencia1} Suppose that the equations (\ref{sistemafinal11})-(\ref%
{sistemafinal12}) are satisfied with $c(t),$ $d(t),$ $f(t)\in
C^{1}([0,\infty ))$ and $c(t)>0,$ for all $t\in \lbrack 0,\infty ).$ If $%
u_{0}\in C_{0}(R),$ then there exists a solution $u\in C([0,\infty
),C_{0}(R))$ for the Cauchy problem 
\begin{align}
& u_{t}=u_{xx}-\left( g\left( t\right) -c\left( t\right) x\right)
u_{x}+\left( d\left( t\right) +f(t)x-(b(t)-c_{0})x^{2}\right)
u+h(x,t)\left\vert u\right\vert ^{p}u,  \label{prob1} \\
& u(x,0)=u_{0}(x)\mu ^{-1/2}(0)e^{\alpha (0)x^{2}+\delta (0)x+\kappa (0)},
\label{prob2}
\end{align}%
\noindent with $p>0,$ $x\in R,$ $t>0,$ when $h(x,t)$ has the form 
\begin{equation}
h(x,t)=-\mu ^{p/2}(t)\exp \left\{ -p(\alpha (t)x^{2}+\delta (t)x+\kappa
(t))\right\} .
\end{equation}
\end{proposition}

\bigskip 
\begin{table}[tbp]
\caption{Examples of variable coefficient Burgers equations with explicit
solutions.}{\small {\fontfamily{cmr10}%
\selectfont{
\begin{tabular}{|c|c|}
\hline
\# & \textbf{Generalized Burgers Equation} \\ \hline
1 & $v_{t}-(vv_{x}-v_{xx})=-\frac{1}{4}a_{0}b_{0}\exp \left\{ -\frac{a_{0}b_{0}t^{m+n+1}}{m+n+1}\right\}
t^{m+n}+(b_{0}mt^{m-1}-a_{0}b_{0}^{2}t^{2m+n})x$ \\ \hline
2 & $v_{t}-(vv_{x}-v_{xx})=\sin {t}$ \\ \hline
3 & $v_{t}-a_{0}t^{n}(vv_{x}-v_{xx})=\frac{1}{4}a_{0}c_{0}\exp \left\{ \frac{a_{0}c_{0}t^{n+1}}{n+1}\right\} t^{n}-a_{0}c_{0}^{2}t^{n}x$ \\ \hline
4 & $v_{t}-(vv_{x}-v_{xx})=(c_{0}kt^{k-1}-c_{0}^{2}t^{2k})x$ \\ \hline
5 & $v_{t}-a_{0}t^{n-1}(vv_{x}-v_{xx})=-2a_{0}b_{0}\exp \left\{ -\frac{a_{0}b_{0}t^{n}}{n}\right\} t^{n-1}-a_{0}b_{0}^{2}t^{n-1}x$ \\ \hline
6 & $v_{t}-(vv_{x}-v_{xx})=-\lambda _{0}^{2}x+\cosh {t}$ \\ \hline
7 & 
\begin{tabular}{lll}
$v_{t}-b_{0}\exp \left\{ \mu _{0}t\right\} (vv_{x}-v_{xx})$ & $=-2a_{0}b_{0}\exp \left\{ (\lambda _{0}+\mu _{0})t-\frac{a_{0}b_{0}\exp
\left\{ (\lambda _{0}+\mu _{0})t\right\} }{\lambda _{0}+\mu _{0}}\right\} $
&  \\ 
& $+(a_{0}\lambda _{0}\exp \left\{ \lambda _{0}t\right\} -a_{0}^{2}b_{0}\exp
\left\{ (2\lambda _{0}+\mu _{0})t\right\} )x$ & 
\end{tabular}
\\ \hline
8 & $v_{t}-a_{0}\exp \left\{ \lambda _{0}t\right\}
(vv_{x}-v_{xx})=(b_{0}nt^{n-1}-a_{0}b_{0}^{2}\exp \left\{ \lambda
_{0}t\right\} t^{2n})x$ \\ \hline
9 & $v_{t}-a_{0}t^{n}(vv_{x}-v_{xx})=(\lambda _{0}b_{0}\exp \left\{ \lambda
_{0}t\right\} -a_{0}b_{0}^{2}\exp \left\{ 2\lambda _{0}t\right\} t^{n})$ \\ 
\hline
10 & 
\begin{tabular}{lll}
$v_{t}+(\lambda _{0}-a_{0}\sinh ^{2}{\lambda _{0}t})(vv_{x}-v_{xx})$ & $=2\exp \left\{ -\frac{a_{0}\cosh {2\lambda _{0}t}}{4\lambda _{0}}\right\}
\cosh {\lambda _{0}t}(\lambda _{0}-a_{0}\sinh ^{2}{\lambda _{0}t})$ &  \\ 
& $+(\lambda _{0}-a_{0}-a_{0}\sinh ^{2}{\lambda _{0}t})x$ & 
\end{tabular}
\\ \hline
\end{tabular}\label{tabla}
}}\newline
\ }
\end{table}

\section{\protect\bigskip Explicit Solutions for a Generalized Burgers
Equation with variable coefficients}

In this Section, we study explicit solutions with multiparameters through a
new transformation (using Riccati systems) for the\ following generalized
Burgers equation with variable coefficient (\ref{GBE0}): 
\begin{equation}
v_{t}+4a(t)(vv_{x}-v_{xx})=-b(t)x+f(t).  \label{VCBE2}
\end{equation}

The case $a(t)=1/4$ and $f(t)=0$ was studied by Eule and Friedrich in \cite%
{Eule:Friedrich}. They also considered random terms. Our approach is similar
to \cite{Salas} and \ \cite{Buyukasik:pashaev} with the advantage of using
our Riccati system that allows multiparameters. As an application of our
multiparameter approach we present a new symmetry for Burgers equation, see
proposition 2. Abundant families can be produced with explicit solutions and
we present some examples on Tables 3 and 4.

\subsection{\protect\bigskip Classical solutions of the Burgers equation}

We recall some classical solutions of BE 
\begin{equation}
u_{\tau }+u_{\xi }u-u_{_{\xi \xi }}=0.  \label{BE}
\end{equation}

The following are well-known solutions for the equation (\ref{BE}) \cite%
{Buyukasik:pashaev}:

a) Shock solitary wave solution%
\begin{equation}
u(\xi ,\tau )=c-A\tanh \left[ \frac{A}{2}(\xi -c\tau +c_{0})\right]
\label{u7}
\end{equation}

b) Triangular wave solution 
\begin{equation}
u(\xi ,\tau )=\frac{1}{\sqrt{2\pi \tau }}\left( \frac{\left( e^{A}-1\right)
e^{-\xi ^{2}/2\tau }}{1+\frac{1}{2}\left( e^{A}-1\right) erfc[\xi /\sqrt{%
2\tau }]}\right)  \label{u8}
\end{equation}

corresponding to initial condition $u(\xi ,0)=A\delta (\xi ),$ $A$ is
constant, $\delta $ is the delta-Dirac and $erfc[a]=2/\sqrt{\pi }%
\int_{a}^{\infty }e^{-x^{2}}dx.$

\begin{table}[tbp]
\caption{Solutions of variable coefficient Burgers equations of Table 3 of
the form (\protect\ref{substitution2}). $u_{i}$ to be chosen from (\protect
\ref{u7})- (\protect\ref{u10}). }{\small {\fontfamily{cmr10}%
\selectfont{
\begin{tabular}{|c|c|}
\hline
\# & \textbf{Solutions of Burgers Equations} \\ \hline
1 & 
\begin{tabular}{l}
$\frac{1}{8}\exp \left\{ \frac{a_{0}b_{0}t^{m+n+1}}{m+n+1}\right\}
+b_{0}t^{m}x+\exp \left\{ \frac{a_{0}b_{0}t^{m+n+1}}{m+n+1}\right\} \times $
\\ 
$u_{i}\left[ \frac{a_{0}t^{n+1}}{8(n+1)}\exp \left\{ \frac{a_{0}b_{0}t^{m+n+1}}{m+n+1}\right\} x,\right. \left. \frac{2^{-\frac{n+1}{m+n+1}}a_{0}t^{n+1}\left( -\frac{a_{0}b_{0}t^{m+n+1}}{m+n+1}\right) ^{-\frac{n+1}{m+n+1}}\Gamma \left[ \frac{n+1}{m+n+1},-\frac{2a_{0}b_{0}t^{m+n+1}}{m+n+1}\right] }{m+n+1}\right] $\end{tabular}
\\ \hline
2 & $-\frac{x}{t}+\frac{\sin {t}-t\cos {t}}{t}+\frac{1}{t}u_{i}\left[ \frac{x}{t}-\frac{\sin {t}}{t},\frac{1}{t}\right] $ \\ \hline
3 & 
\begin{tabular}{l}
$\frac{1}{8}\exp \left\{ \frac{a_{0}c_{0}t^{n+1}}{n+1}\right\} -c_{0}x+\exp
\left\{ -\frac{a_{0}c_{0}t^{n+1}}{n+1}\right\} \times $ \\ 
$u_{i}\left[ \frac{a_{0}t^{n+1}}{8(n+1)}+\exp \left\{ -\frac{a_{0}c_{0}t^{n+1}}{n+1}\right\} x,\frac{\exp \left\{ -2\frac{a_{0}c_{0}t^{n+1}}{n+1}\right\} }{2c_{0}}\right] $\end{tabular}
\\ \hline
4 & $c_{0}t^{k}x+\exp \left\{ \frac{c_{0}t^{k+1}}{k+1}\right\} \times u_{i}\left[ \exp \left\{ \frac{c_{0}t^{k+1}}{k+1}\right\} x,\frac{2^{-\frac{1}{k+1}}t\left( -\frac{c_{0}t^{k+1}}{k+1}\right) ^{-\frac{1}{k+1}}\Gamma \left[ 
\frac{1}{k+1},-\frac{2c_{0}t^{k+1}}{k+1}\right] }{k+1}\right] $ \\ \hline
5 & $\exp \left\{ -\frac{a_{0}b_{0}t^{n}}{n}\right\} +b_{0}x+\exp \left\{ 
\frac{a_{0}b_{0}t^{n}}{n}\right\} u_{i}\left[ \frac{a_{0}t^{n}}{n}+\exp
\left\{ \frac{a_{0}b_{0}t^{n}}{n}\right\} x,-\frac{\exp \left\{ 2\frac{a_{0}b_{0}t^{n}}{n}\right\} }{2b_{0}}\right] $ \\ \hline
6 & $-\lambda _{0}x+\frac{\lambda _{0}\cosh {t}-\sinh {t}}{\lambda _{0}^{2}-1}+\exp \left\{ -\lambda _{0}t\right\} u_{i}\left[ \exp \left\{ -\lambda
_{0}t\right\} x-\frac{\exp \left\{ -\lambda _{0}t\right\} \cosh {t}}{\lambda
_{0}^{2}-1},\frac{\exp \left\{ -2\lambda _{0}t\right\} }{2\lambda _{0}}\right] $ \\ \hline
7 & 
\begin{tabular}{l}
$\exp \left\{ -\frac{a_{0}b_{0}\exp \left\{ (\lambda _{0}+\mu _{0})t\right\} 
}{\lambda _{0}+\mu _{0}}\right\} +a_{0}\exp \left\{ \lambda _{0}t\right\}
x+\exp \left\{ \frac{a_{0}b_{0}\exp \left\{ (\lambda _{0}+\mu _{0})t\right\} 
}{\lambda _{0}+\mu _{0}}\right\} \times $ \\ 
$u_{i}\left[ \exp \left\{ -\frac{a_{0}b_{0}\exp \left\{ (\lambda _{0}+\mu
_{0})t\right\} }{\lambda _{0}+\mu _{0}}\right\} x+\frac{b_{0}\exp \left\{
\mu _{0}t\right\} }{\mu _{0}},-b_{0}\int \exp \left\{ \mu _{0}t+2\frac{a_{0}b_{0}\exp \left\{ (\lambda _{0}+\mu _{0})t\right\} }{\lambda _{0}+\mu
_{0}}\right\} dt\right] $\end{tabular}
\\ \hline
8 & 
\begin{tabular}{l}
$b_{0}t^{n}x+\exp \left\{ -a_{0}b_{0}t^{n+1}(-\lambda _{0}t)^{-(n+1)}\Gamma \left[ n+1,-\lambda _{0}t\right] \right\} \times $ \\ 
$u_{i}\left[ \exp \left\{ -a_{0}b_{0}t^{n+1}(-\lambda _{0}t)^{-(n+1)}\Gamma \left[ n+1,-\lambda _{0}t\right] \right\} x,\right. $ \\ 
$\left. a_{0}\int \exp \left\{ \lambda _{0}t-a_{0}b_{0}t^{n+1}(-\lambda
_{0}t)^{-(n+1)}\Gamma \left[ n+1,-\lambda _{0}t\right] \right\} dt\right] $\end{tabular}
\\ \hline
9 & 
\begin{tabular}{l}
$b_{0}\exp \left\{ \lambda _{0}t\right\} x+\exp \left\{
-a_{0}b_{0}t^{n+1}(-\lambda _{0}t)^{-(n+1)}\Gamma \left[ n+1,-\lambda _{0}t\right] \right\} \times $ \\ 
$u_{i}\left[ \exp \left\{ -a_{0}b_{0}t^{n+1}(-\lambda _{0}t)^{-(n+1)}\Gamma \left[ n+1,-\lambda _{0}t\right] \right\} x,\right. $ \\ 
$\left. -a_{0}\int \exp \left\{ -2a_{0}b_{0}t^{n+1}(-\lambda
_{0}t)^{-(n+1)}\Gamma \left[ n+1,-\lambda _{0}t\right] \right\} t^{n}dt\right] $\end{tabular}
\\ \hline
10 & 
\begin{tabular}{l}
$\exp \left\{ -\frac{a_{0}\cosh {2\lambda _{0}t}}{4\lambda _{0}}\right\}
\sinh {\lambda _{0}t}+\cosh {\lambda _{0}t}x+\exp \left\{ \frac{a_{0}\cosh {2\lambda _{0}t}}{4\lambda _{0}}\right\} \csch{\lambda_{0}t}\times $ \\ 
$u_{i}\left[ \exp \left\{ \frac{a_{0}\cosh {2\lambda _{0}t}}{4\lambda _{0}}\right\} \csch{\lambda_{0}}x+2\left( \frac{a_{0}\sinh {2\lambda _{0}t}}{8\lambda _{0}}-\frac{a_{0}t}{4}-\frac{\lambda _{0}t}{2}\right) ,\right. $ \\ 
$\left. \int \exp \left\{ \frac{a_{0}\cosh {2\lambda _{0}t}}{4\lambda _{0}}\right\} \csch^{2}{\lambda _{0}t}(\lambda _{0}-a_{0}\sinh ^{2}{\lambda _{0}t})dt\right] $\end{tabular}
\\ \hline
\end{tabular}\label{tabla}
}}\newline
\ }
\end{table}

\bigskip

c) N-wave solution%
\begin{equation}
u(\xi ,\tau )=\frac{\xi }{\tau }\left( \frac{\sqrt{a/\tau }e^{-\xi
^{2}/4\tau }}{1+\sqrt{a/\tau }e^{-\xi ^{2}/4\tau }}\right) ,\tau >0.
\label{u9}
\end{equation}

d) Quotient of Kampè de Feriet polynomials%
\begin{equation}
u_{k}(\xi ,\tau )=-\frac{\sum_{m=1}^{k}ma_{m}H_{m-1}(\xi ,\tau /2)}{%
\sum_{m=0}^{k}a_{m}H_{m}(\xi ,\tau /2)},k=1,2,3,....  \label{u10}
\end{equation}

\bigskip where Kampè de Feriet polynomials are defined by 
\begin{equation}
H_{m}(\xi ,\tau /2)=m!\sum_{k=0}^{[m/2]}\frac{\left( \tau /2\right) ^{k}}{%
k!(m-2k)!}\xi ^{m-2k},\qquad H_{m}(\xi ,0)=\xi ^{m},  \label{Kampe Poly}
\end{equation}

with $[m/2]=m/2$ for even $m,$ and $[m/2]=(m-1)/2$ for odd $m,$ see \cite%
{Buyukasik:pashaev}.\bigskip {\small \ }

\subsection{Soliton solutions for VCBE (\protect\ref{VCBE2}) through Riccati
equations and similarity transformations}

\begin{proposition}
\bigskip If the following second-order differential equation can be solved
explicitly\ 
\begin{equation}
\mu ^{\prime \prime }(t)-\frac{a^{\prime }(t)}{a(t)}\mu ^{\prime
}(t)+4a(t)b(t)\mu (t)=0,  \label{char equ}
\end{equation}%
\begin{equation}
\mu _{0}\left( 0\right) =0,\quad \mu _{0}^{\prime }\left( 0\right) =2a\left(
0\right) \neq 0\qquad \mu _{1}\left( 0\right) \neq 0,\quad \mu _{1}^{\prime
}\left( 0\right) =0,  \label{standarddata}
\end{equation}

then the Cauchy initial value problem for the generalized Burgers equation%
\begin{equation}
v_{t}+4a(t)(vv_{x}+Lv_{xx})=-b(t)x+f(t)  \label{GBE}
\end{equation}

can be reduced to the standard Burgers equation%
\begin{equation}
u_{\tau }+Lu_{\xi \xi }+u_{\xi }u=0  \label{BE1}
\end{equation}

through the \textbf{multiparameter} substitution%
\begin{equation}
v(x,t)=\alpha (t)x+\delta (t)+\beta (t)u(\xi ,\tau ),  \label{substitution2}
\end{equation}

where $\xi =\beta (t)x+2\varepsilon (t)$ and $\tau (t)=4\gamma (t)$ and $\mu
,\alpha ,\beta ,\gamma ,\delta $ and $\varepsilon $ are given explicitly by (%
\ref{MKernel})-(\ref{ap8}) with%
\begin{equation}
\alpha _{0}\left( t\right) =-\frac{1}{4a\left( t\right) }\frac{\mu
_{0}^{\prime }\left( t\right) }{\mu _{0}\left( t\right) },\quad \beta
_{0}\left( t\right) =\frac{1}{\mu _{0}\left( t\right) },\quad \gamma
_{0}\left( t\right) =-\frac{1}{2\mu _{1}\left( 0\right) }\frac{\mu
_{1}\left( t\right) }{\mu _{0}\left( t\right) },  \label{alpha0BE}
\end{equation}%
\begin{equation}
\delta _{0}\left( t\right) =\frac{1}{\mu _{0}\left( t\right) }\ \
\int_{0}^{t}f\left( s\right) \mu _{0}\left( s\right) \ ds,
\end{equation}%
\begin{align}
\varepsilon _{0}\left( t\right) & =-\frac{2a\left( t\right) \delta
_{0}\left( t\right) }{\mu _{0}^{\prime }\left( t\right) }-8\int_{0}^{t}\frac{%
a^{2}\left( s\right) b(s)}{\left( \mu _{0}^{\prime }\left( s\right) \right)
^{2}}\left( \mu _{0}\left( s\right) \delta _{0}\left( s\right) \right) \ ds
\\
& \quad +2\int_{0}^{t}\frac{a\left( s\right) f(s)}{\mu _{0}^{\prime }\left(
s\right) }\ ds,  \notag
\end{align}%
\begin{align}
\kappa _{0}\left( t\right) & =-\frac{a\left( t\right) \mu _{0}\left(
t\right) }{\mu _{0}^{\prime }\left( t\right) }\delta _{0}^{2}\left( t\right)
-4\int_{0}^{t}\frac{a^{2}\left( s\right) b(s)}{\left( \mu _{0}^{\prime
}\left( s\right) \right) ^{2}}\left( \mu _{0}\left( s\right) \delta
_{0}\left( s\right) \right) ^{2}\ ds  \label{kappa0BE} \\
& \quad +2\int_{0}^{t}\frac{a\left( s\right) }{\mu _{0}^{\prime }\left(
s\right) }\left( \mu _{0}\left( s\right) \delta _{0}\left( s\right) \right)
f\left( s\right) \ ds  \notag
\end{align}%
with $\delta _{0}\left( 0\right) =g_{0}\left( 0\right) /\left( 2a\left(
0\right) \right) ,$ $\varepsilon _{0}\left( 0\right) =-\delta _{0}\left(
0\right) ,$ $\kappa _{0}\left( 0\right) =0.$ As an application, we can
present the following symmetry for the Burgers equation. If $u$ \ is a
solution of classical Burgers equation (\ref{u7})-(\ref{u10}), then (\ref%
{substitution2}) becomes 
\begin{eqnarray}
v(x,t) &=&\left( -\frac{1}{2t}-\frac{1}{t\left( \alpha \left( 0\right)
t-1\right) }\right) x-\frac{\delta \left( 0\right) }{\alpha \left( 0\right)
t-1}  \notag \\
&&-\frac{\beta \left( 0\right) }{\alpha \left( 0\right) t-1}u\left( -\frac{%
\beta \left( 0\right) }{\alpha \left( 0\right) t-1}x+\varepsilon \left(
0\right) -\frac{t\beta \left( 0\right) \delta \left( 0\right) }{2\left(
\alpha \left( 0\right) t-1\right) },\gamma \left( 0\right) -\frac{\beta
^{2}\left( 0\right) t}{4\left( \alpha \left( 0\right) t-1\right) }\right) ,
\end{eqnarray}%
and it is is also a (multiparameter) solution with the arbitrary data $%
\alpha \left( 0\right) ,$ $\beta \left( 0\right) \neq 0,$ $\gamma (0),$ $%
\delta (0),$ $\varepsilon (0)$ and $\kappa (0)$
\end{proposition}

\textbf{Proof:} We look for a solution of the form (\ref{substitution2}).
Using substitution (\ref{substitution2})\ in (\ref{GBE}) we obtain

\begin{eqnarray*}
v_{t} &=&\dot{\alpha}x+\dot{\delta}+\dot{\beta}u+\beta \lbrack u_{\xi }(\dot{%
\beta}x+2\dot{\varepsilon})+4\dot{\tau}u_{\tau }], \\
aLv_{xx} &=&a\beta ^{3}Lu_{\xi \xi }, \\
avv_{x} &=&a\alpha ^{2}x+a\alpha \beta ^{2}xu_{\xi }+a\delta \alpha +a\delta
\beta ^{2}u_{\xi }+a\beta \alpha u+a\beta ^{3}u_{\xi }u.
\end{eqnarray*}

Forcing $a,$ $b$ and $f$ to satisfy the Riccati system

\begin{align}
& \frac{d\alpha }{dt}+b+4a\alpha ^{2}=0,  \label{alpha(t)} \\
& \frac{d\beta }{dt}+4a\alpha \beta =0,  \label{beta(t)} \\
& \frac{d\gamma }{dt}=a\beta ^{2},  \label{gamma(t)} \\
& \frac{d\delta }{dt}+4a\alpha \delta =f,  \label{delta(t)} \\
& \frac{d\varepsilon }{dt}+2a\delta \beta =0,  \label{epsilon(t)}
\end{align}

equation (\ref{GBE}) would be reduced to\ the standard BE (\ref{BE1}).
Making the substitution

\begin{equation}
\frac{\mu ^{\prime }}{2\mu }+2a\alpha =0  \label{mu(t)}
\end{equation}%
into (\ref{alpha(t)}) we obtain (\ref{char equ}).

\begin{remark}
The solutions presented in Table 4 for the families in Table 3 are of the
form (\ref{substitution2}) where for convenience of the presentation $\alpha
,\beta ,\gamma ,\delta $ and $\varepsilon $ satisifies (\ref{alpha0})-(\ref%
{kappa0}), but we would like to emphasize that the most general
multiparamters solutions are given by (\ref{substitution2}) and $\alpha
,\beta ,\gamma ,\delta $ and $\varepsilon $ satisfies (\ref{AKernel})-(\ref%
{kappa0}).
\end{remark}

\textbf{Conclusion.} \emph{The Riccati equations have played an important
role in explicit solutions for Fisher and Burgers equations and other
nonlinear PDES (see \cite{Feng}, \cite{Gomez}, \cite{Lan:Lop:Sus}, \cite%
{Sopho}, \cite{Suazo:Sus:Ve}, \cite{Suazo:Sus:Ve2} and references therein).
Also, solutions for Ermakov systems have been a very useful tool in the
study of dispersive equations (see \cite{Escorcia} for an extended
bibliography) including applications to 1D Bose-Einstein condensate. In this
paper, in order to obtain the main results, we use a fundamental approach
consisting of the use of similarity transformations and the solutions of
Riccati-Ermakov systems with several parameters for the diffusion case. In
general, Fisher-KPP and Burgers equations type with variable coefficients
(vcBE) are not integrable so similarity transformations and other methods
have been extensively applied \cite{Moreau}, \cite{Orlo}, \cite{Salas}, \cite%
{Sopho}, \cite{Xu} and \cite{Eule:Friedrich}, \cite{Sakai} and \cite{Zola}.}

\emph{Also, in this work, inspired by the work of Mahric \cite{Marhic78} on
multiparameter solutions for the linear Schrödinger equation with quadratic
potential, we have established a relationship between solutions with
parameters of Riccati-Ermakov systems and the dynamics of a Fisher-KPP and
Burgers-type equations with variable coefficients. In fact, as a first
application of this multiparameter approach we present explicit solutions
with singularities for selected coefficients using Riccati systems and
similarity transformations from standard solutions of Fisher-KPP and Burgers
equations, see Tables 1-4 for a list of several families. For a generalized
Fisher-KPP, we can obtain traveling wave and rational solutions. And for a
generalized Burgers equation, we can obtain triangular wave and N-wave type
solutions. We explore the dynamics of these solutions across
multi-parameters by means of Riccati-Ermakov systems (see Figures 1-3) as in
previous works for the paraxial wave equation \cite{Ma:Suslov} and nonlinear
Schrödinger equation in \cite{Escorcia}. This work should motivate further
analytical and numerical studies looking to clarify the connections between
the dynamics of variable-coefficient NLS and the solutions of ordinary
differential Riccati-Ermakov systems. }

\emph{In this work we also generalized the transformation for a vcBE
presented in \cite{Sopho} (see our proposition 2) where Langevin equations
and the Hill equation were used to express the transformation. Instead of
these approaches, in proposition 2 we will use what we have called the
Riccati system, extending the results presented in \cite{Buyukasik:pashaev}
and \cite{Salas}; further, our solutions will show multiparameters and we
present a new symmetry for Burgers equation. Tables 3 and 4 shows several
examples of families with explicit solutions.}

\emph{Finally, we have prepared a Mathematica file as supplemental material
verifying the solutions. Also, all the formulas from the appendix have been
verified previously \cite{Suazo:Sus:Ve2}.}

\begin{acknowledgement}
One of the authors (ES) was supported by the Simons Foundation Grant \#
316295, by NSF grant DMS\#1620196 and College of Sciences Research
Enhancement Seed Grants Program at UTRGV.
\end{acknowledgement}

\section{Appendix: Riccati and Ermakov Systems (diffusion case)}

\bigskip The following result can be found at \cite{Suazo:Sus:Ve2}.

\begin{lem}
(Solutions for the Riccati system.) The system of ordinary differential
equations (\ref{Ermakov01})-(\ref{Ermakov06}) with $c_{0}=0$ has the
following multiparameter general solution: 
\begin{eqnarray}
&&\mu \left( t\right) =-2\mu \left( 0\right) \mu _{0}\left( t\right) \left(
\alpha \left( 0\right) +\gamma _{0}\left( t\right) \right) ,  \label{MKernel}
\\
&&\alpha \left( t\right) =\alpha _{0}\left( t\right) -\frac{\beta
_{0}^{2}\left( t\right) }{4\left( \alpha \left( 0\right) +\gamma _{0}\left(
t\right) \right) },  \label{AKernel} \\
&&\beta \left( t\right) =-\frac{\beta \left( 0\right) \beta _{0}\left(
t\right) }{2\left( \alpha \left( 0\right) +\gamma _{0}\left( t\right)
\right) }=\frac{\beta \left( 0\right) \mu \left( 0\right) }{\mu \left(
t\right) }\lambda \left( t\right) ,  \label{BKernel} \\
&&\gamma \left( t\right) =\gamma \left( 0\right) -\frac{\beta ^{2}\left(
0\right) }{4\left( \alpha \left( 0\right) +\gamma _{0}\left( t\right)
\right) },  \label{CKernel}
\end{eqnarray}%
\begin{eqnarray}
\delta \left( t\right) &=&\delta _{0}\left( t\right) -\frac{\beta _{0}\left(
t\right) \left( \delta \left( 0\right) +\varepsilon _{0}\left( t\right)
\right) }{2\left( \alpha \left( 0\right) +\gamma _{0}\left( t\right) \right) 
},  \label{ap6} \\
\varepsilon \left( t\right) &=&\varepsilon \left( 0\right) -\frac{\beta
\left( 0\right) \left( \delta \left( 0\right) +\varepsilon _{0}\left(
t\right) \right) }{2\left( \alpha \left( 0\right) +\gamma _{0}\left(
t\right) \right) },  \label{ap7} \\
\kappa \left( t\right) &=&\kappa \left( 0\right) +\kappa _{0}\left( t\right)
-\frac{\left( \delta \left( 0\right) +\varepsilon _{0}\left( t\right)
\right) ^{2}}{4\left( \alpha \left( 0\right) +\gamma _{0}\left( t\right)
\right) }  \label{ap8}
\end{eqnarray}%
with the arbitrary data $\mu \left( 0\right) ,$ $\alpha \left( 0\right) ,$ $%
\beta \left( 0\right) \neq 0,$ $\gamma (0),$ $\delta (0),$ $\varepsilon (0)$
y $\kappa (0)$ where $\alpha _{0},\beta _{0},\gamma _{0},\delta
_{0},\epsilon _{0}$ and $\kappa _{0}$ are given explicitly for: 
\begin{equation}
\alpha _{0}\left( t\right) =-\frac{1}{4a\left( t\right) }\frac{\mu
_{0}^{\prime }\left( t\right) }{\mu _{0}\left( t\right) }-\frac{d\left(
t\right) }{2a\left( t\right) },  \label{alpha0}
\end{equation}%
\begin{equation}
\beta _{0}\left( t\right) =\frac{\lambda \left( t\right) }{\mu _{0}\left(
t\right) },\quad \lambda \left( t\right) =\exp \left( \int_{0}^{t}\left(
c\left( s\right) -2d\left( s\right) \right) \ ds\right) ,
\end{equation}%
\begin{align}
\gamma _{0}\left( t\right) & =\frac{d\left( 0\right) }{2a\left( 0\right) }-%
\frac{a\left( t\right) \lambda ^{2}\left( t\right) }{\mu _{0}\left( t\right)
\mu _{0}^{\prime }\left( t\right) }-4\int_{0}^{t}\frac{a\left( s\right)
\sigma \left( s\right) \lambda \left( s\right) }{\left( \mu _{0}^{\prime
}\left( s\right) \right) ^{2}}\ ds \\
& =\frac{d\left( 0\right) }{2a\left( 0\right) }-\frac{1}{2\mu _{1}\left(
0\right) }\frac{\mu _{1}\left( t\right) }{\mu _{0}\left( t\right) },
\end{align}%
\begin{equation}
\delta _{0}\left( t\right) =\frac{\lambda \left( t\right) }{\mu _{0}\left(
t\right) }\ \ \int_{0}^{t}\left[ \left( f\left( s\right) +\frac{d\left(
s\right) }{a\left( s\right) }g\left( s\right) \right) \mu _{0}\left(
s\right) +\frac{g\left( s\right) }{2a\left( s\right) }\mu _{0}^{\prime
}\left( s\right) \right] \ \frac{ds}{\lambda \left( s\right) },
\end{equation}%
\begin{align}
\varepsilon _{0}\left( t\right) & =-\frac{2a\left( t\right) \lambda \left(
t\right) }{\mu _{0}^{\prime }\left( t\right) }\delta _{0}\left( t\right)
-8\int_{0}^{t}\frac{a\left( s\right) \sigma \left( s\right) \lambda \left(
s\right) }{\left( \mu _{0}^{\prime }\left( s\right) \right) ^{2}}\left( \mu
_{0}\left( s\right) \delta _{0}\left( s\right) \right) \ ds  \label{Epsilon0}
\\
& \quad +2\int_{0}^{t}\frac{a\left( s\right) \lambda \left( s\right) }{\mu
_{0}^{\prime }\left( s\right) }\left[ f\left( s\right) +\frac{d\left(
s\right) }{a\left( s\right) }g\left( s\right) \right] \ ds,  \notag
\end{align}%
\begin{align}
\kappa _{0}\left( t\right) & =-\frac{a\left( t\right) \mu _{0}\left(
t\right) }{\mu _{0}^{\prime }\left( t\right) }\delta _{0}^{2}\left( t\right)
-4\int_{0}^{t}\frac{a\left( s\right) \sigma \left( s\right) }{\left( \mu
_{0}^{\prime }\left( s\right) \right) ^{2}}\left( \mu _{0}\left( s\right)
\delta _{0}\left( s\right) \right) ^{2}\ ds  \label{kappa0} \\
& \quad +2\int_{0}^{t}\frac{a\left( s\right) }{\mu _{0}^{\prime }\left(
s\right) }\left( \mu _{0}\left( s\right) \delta _{0}\left( s\right) \right) %
\left[ f\left( s\right) +\frac{d\left( s\right) }{a\left( s\right) }g\left(
s\right) \right] \ ds  \notag
\end{align}%
with $\sigma =ab+cd-d^{2}+d/2(\dot{a}/a-\dot{d}/d),$ $\delta _{0}\left(
0\right) =g_{0}\left( 0\right) /\left( 2a\left( 0\right) \right) ,$ $%
\varepsilon _{0}\left( 0\right) =-\delta _{0}\left( 0\right) ,$ $\kappa
_{0}\left( 0\right) =0.$ Here $\mu _{0}$ and $\mu _{1}$ represent the
fundamental solution of the characteristic equation (\ref{Ermakov-P})
subject to the initial conditions $\mu _{0}(0)=0$, $\mu _{0}^{\prime
}(0)=2a(0)\neq 0$ and $\mu _{1}(0)\neq 0$, $\mu _{1}^{\prime }(0)=0$.{\small %
\ }

\begin{remark}
We can easily replace equation (\ref{gamma}) by 
\begin{equation}
\dot{\gamma}(t)+l_{0}a(t)\beta ^{2}(t)=0,\quad l_{0}=\pm 1,
\end{equation}%
with solution 
\begin{equation}
\gamma \left( t\right) =l_{0}\gamma \left( 0\right) -\frac{l_{0}\beta
^{2}\left( 0\right) }{4\left( \alpha \left( 0\right) +\gamma _{0}\left(
t\right) \right) },\quad l_{0}=\pm 1,
\end{equation}%
obtaining a larger variety of solutions.
\end{remark}
\end{lem}

The following result can be found at \cite{Suazo:Sus:Ve2}.

\begin{lem}
(Solutions for the Ermakov System.) The system of ordinary differential
equations (\ref{Ermakov01})-(\ref{Ermakov06}) with $c_{0}=1$ has the
following multiparameter general solution: 
\begin{align}
\mu (t)& =\mu _{0}\mu (0)\sqrt{4(\gamma _{0}+\alpha (0))^{2}-\beta ^{4}(0)},
\label{Erma1} \\
\alpha (t)& =\alpha _{0}-\dfrac{\beta _{0}^{2}(\gamma _{0}+\alpha (0))}{%
4(\gamma _{0}+\alpha (0))^{2}-\beta ^{4}(0)},  \label{Erma2} \\
\beta (t)& =\dfrac{\beta (0)\beta _{0}}{\sqrt{4(\gamma _{0}+\alpha
(0))^{2}-\beta ^{4}(0)}},  \label{Erma3} \\
\gamma (t)& =\gamma (0)-\dfrac{1}{4}\ln \left[ \dfrac{(\gamma _{0}+\alpha
(0))+\frac{1}{2}\beta ^{2}(0)}{(\gamma _{0}+\alpha (0))-\frac{1}{2}\beta
^{2}(0)}\right] ,  \label{Erma4} \\
\delta (t)& =\delta _{0}+\beta _{0}\dfrac{\varepsilon (0)\beta
^{3}(0)-2(\gamma _{0}+\alpha (0))(\varepsilon _{0}+\delta (0))}{4(\gamma
_{0}+\alpha (0))^{2}-\beta ^{4}(0)},  \label{Erma5} \\
\varepsilon (t)& =\dfrac{\beta (0)(\delta (0)+\varepsilon _{0})-\varepsilon
(0)(\gamma _{0}+\alpha (0))}{\sqrt{4(\gamma _{0}+\alpha (0))^{2}-\beta
^{4}(0)}},  \label{Erma6} \\
\kappa (t)& =\kappa _{0}+\kappa (0)+\dfrac{\beta ^{3}(0)\varepsilon
(0)(\varepsilon _{0}+\delta (0))}{4(\gamma _{0}+\alpha (0))^{2}-\beta ^{4}(0)%
}  \label{Erma7} \\
& -\dfrac{(\gamma _{0}+\alpha (0))\left[ \beta ^{2}(0)\varepsilon
^{2}(0)+(\varepsilon _{0}+\delta (0))^{2}\right] }{4(\gamma _{0}+\alpha
(0))^{2}-\beta ^{4}(0)}  \notag
\end{align}%
where $\alpha _{0},\beta _{0},\gamma _{0},\delta _{0},\epsilon _{0}$ and $%
\kappa _{0}$ are given explicitly for (\ref{alpha0})-(\ref{kappa0})\ \ with
the arbitrary initial conditions $\mu (0)>0,$ $\alpha (0),$ $\beta (0)\neq
0, $ $\gamma (0),$ $\kappa (0),$ $\delta (0),$ $\varepsilon (0).$
\end{lem}

\end{document}